\documentclass[a4paper,canadian]{lipics-v2018}
\usepackage{lmodern}
\usepackage{microtype}
\usepackage{etoolbox}
\usepackage{thm-restate}
\usepackage{wrapfig}

\graphicspath{{./images/}}

\theoremstyle{plain}
\newtheorem{prop}{Proposition}
\newtheorem{lem}[prop]{Lemma}

\newtheorem{defi}[prop]{Definition}
\newtheorem{ex}[prop]{Example}

\newtheorem{cor}[prop]{Corollary}

\newcommand{\short}[1]{}
\newcommand{\full}[1]{#1}

\newcommand{\cbox}[1]{\vspace{0.2cm}\noindent
  \fbox{\parbox{.97\textwidth}{#1}}\vspace{0.2cm}}

\title{Updating Probabilistic Knowledge on Condition/Event Nets using
  Bayesian Networks}

\author{Benjamin Cabrera}{University of Duisburg-Essen}{benjamin.cabrera@uni-due.de}{}{}
\author{Tobias Heindel}{University of Hawaii}{heindel@hawaii.edu}{}{}
\author{Reiko Heckel}{University of Leicester}{rh122@leicester.ac.uk}{}{}
\author{Barbara König}{University of Duisburg-Essen}{barbara\_koenig@uni-due.de}{}{}

\authorrunning{B. Cabrera and T. Heindel and R. Heckel and B. König}

\Copyright{Benjamin Cabrera and Tobias Heindel and Reiko Heckel and Barbara König}

\subjclass{
  ,
}


\funding{Research partially supported by the Deutsche Forschungsgemeinschaft  (DFG) under grant No. GRK 2167, Research Training Group ``User-Centred Social Media''.}

\EventEditors{Sven Schewe and Lijun Zhang}
\EventNoEds{2}
\EventLongTitle{29th International Conference on Concurrency Theory
(CONCUR 2018)}
\EventShortTitle{CONCUR 2018}
\EventAcronym{CONCUR}
\EventYear{2018}
\EventDate{September 4--7, 2018}
\EventLocation{Beijing, China}
\EventLogo{}
\SeriesVolume{118}
\ArticleNo{27}
\nolinenumbers 
\full{\hideLIPIcs}  

\usepackage{amsmath}
\usepackage{amssymb}
\usepackage{amsfonts}
\usepackage{mathtools}

\usepackage[disable]{todonotes}

\usepackage{pgf}
\usepackage{tikz}
\usetikzlibrary{matrix,arrows,shapes,automata,backgrounds,petri}

\newcommand{\leftdot}[1]{\prescript{\bullet}{}{#1}}
\newcommand{\rightdot}[1]{{#1}^\bullet}

\renewcommand{\phi}{\varphi}
\renewcommand{\epsilon}{\varepsilon}
\newcommand{\prob}{\mathbb{P}}

\newcommand{\assertop}{\mathrm{ass}}
\newcommand{\nassertop}{\mathrm{nas}}
\newcommand{\setop}{\mathrm{set}}
\newcommand{\successop}{\mathrm{success}}
\newcommand{\failop}{\mathrm{fail}}

\newcommand{\markt}[1]{\mathcal{M}[\overset{t}{\Rightarrow}#1]}
\newcommand{\marktnpre}{\mathcal{M}[\overset{t}{\not\Rightarrow}_\textit{pre}]}
\newcommand{\marktnpost}{\mathcal{M}[\overset{t}{\not\Rightarrow}_\textit{post}]}

\newcommand{\bitvec}[1]{\mathbf{#1}}
\newcommand{\id}{\mathrm{id}}
\newcommand{\out}{\mathrm{out}}
\newcommand{\ev}{e}

\usepackage{tabularx}
\usepackage{xspace}

\makeatletter
\newcommand{\gettikzxy}[3]{%
  \tikz@scan@one@point\pgfutil@firstofone#1\relax
  \edef#2{\the\pgf@x}%
  \edef#3{\the\pgf@y}%
}
\newcommand{\gettikzx}[2]{%
  \tikz@scan@one@point\pgfutil@firstofone#1\relax
  \edef#2{\the\pgf@x}%
}
\newcommand{\gettikzy}[2]{%
  \tikz@scan@one@point\pgfutil@firstofone#1\relax
  \edef#2{\the\pgf@y}%
}
\makeatother

\newcommand{\Rcell}[4][1]{%
  \begin{scope}[very thick]
    \dlracell[#1]{1}{1}{#4}{0}{0}
  \end{scope}
  \node[anchor=south,outer sep=1.5pt,inner sep=1pt] at (lastL0) {\ensuremath{#2}};
  \node[anchor=south,outer sep=1.5pt,inner sep=1pt] at (lastR0) {\ensuremath{#3}};
}

\newcommand{\lracell}[6][1]{%
  \ifstrequal{#2}{0}{%
    \phantom{\draw (0,.5) -- +(.25*#1,0);}
  }{%
  \pgfmathtruncatemacro\topx{#5+#2-1}
  \foreach \x in {#5,...,\topx} {
    \begin{scope}[shift={(0,.5)}]
      \draw (0,\x) -- (.25*#1,\x);
    \end{scope}
    \coordinate (lastL\x) at (0,\x+.5);
  }}
\ifstrequal{#3}{0}{%
  \phantom{\draw (.75*#1,.5) -- (#1,.5);}
}{%
  \pgfmathtruncatemacro\topx{#6+#3-1}
  \foreach \x in {#6,...,\topx} {
    \begin{scope}[shift={(0,.5)}]
      \draw (.75*#1,\x) -- (#1,\x);
    \end{scope}
    \coordinate (lastR\x) at (#1,\x+.5);
  }}
\ifthenelse{\equal{#4}{}}%
{}%
{\node%
  [anchor=base,outer sep=1pt,inner sep=1pt,draw,fill=white,thick] %
  (thelabel) at ([yshift=1ex].5*#1,.25+#5/2+#6/2+#2/4+#3/4-.5) 
  {\(#4\)};}
\begin{scope}[shift={(0,.5)}]
  \pgfmathsetmacro{\ytopleft}{max(#2+#5-1,%
    .25+#5/2+#6/2+#2/4+#3/4-.75%
  )}
\pgfmathsetmacro{\ytopright}{max(#3+#6-1,%
  .25+#5/2+#6/2+#2/4+#3/4-.75%
  )}
    \draw[fill=white] %
    (.25*#1,.75+#5-1) 
    -- (.25*#1,\ytopleft+.25) 
    -- (.75*#1,\ytopright+.25) 
    -- (.75*#1,+#6-1+.75) 
    -- cycle    ;
  \end{scope}
  \ifthenelse{\equal{#4}{}}%
{}%
{  \node%
  [anchor=base,outer sep=1pt,inner sep=1pt,draw=none,fill=white] %
  (thelabel) at ([yshift=1ex].5*#1,.25+#5/2+#6/2+#2/4+#3/4-.5) 
  {\(#4\)};}
}

\newcommand{\dlracell}[6][1]{%
  \ifstrequal{#2}{0}{%
    \phantom{\draw (0,.5) -- +(.25*#1,0);}
  }{%
  \pgfmathtruncatemacro\topx{#5+#2-1}
  \foreach \x in {#5,...,\topx} {
    \begin{scope}[shift={(0,.5)}]
      \draw[double] (0,\x) -- (.25*#1,\x);
    \end{scope}
    \coordinate (lastL\x) at (0,\x+.5);
  }}
\ifstrequal{#3}{0}{%
  \phantom{\draw[double] (.75*#1,.5) -- (#1,.5);}
}{%
  \pgfmathtruncatemacro\topx{#6+#3-1}
  \foreach \x in {#6,...,\topx} {
    \begin{scope}[shift={(0,.5)}]
      \draw[double] (.75*#1,\x) -- (#1,\x);
    \end{scope}
    \coordinate (lastR\x) at (#1,\x+.5);
  }}
\ifthenelse{\equal{#4}{}}%
{}%
{\node%
  [anchor=base,outer sep=0pt,inner
  sep=0pt,draw=none,fill=none,thick,line width=1pt,draw=none] %
  (thelabel) at ([yshift=1.5ex].5*#1,.25+#5/2+#6/2+#2/4+#3/4-.5) 
  {\(#4\)};
}
\begin{scope}[shift={(0,.5)}]
  \pgfmathsetmacro{\ytopleft}{max(#2+#5-1,%
    .25+#5/2+#6/2+#2/4+#3/4-.75%
  )}
\pgfmathsetmacro{\ytopright}{max(#3+#6-1,%
  .25+#5/2+#6/2+#2/4+#3/4-.75%
  )}
    \draw[fill=white,semithick] %
    (.25*#1,.75+#5-1) 
    -- (.25*#1,\ytopleft+.25) 
    -- (.75*#1,\ytopright+.25) 
    -- (.75*#1,+#6-1+.75) 
    -- cycle    ;
  \end{scope}
  \ifthenelse{\equal{#4}{}}%
{}%
{
   \node%
   [anchor=base,outer sep=1pt,inner sep=.5pt,draw=none,fill=none] %
   (thelabel) at ([yshift=1.5ex].5*#1,.25+#5/2+#6/2+#2/4+#3/4-.5) 
   {\(#4\)};}
}

\newcommand{\cmp}[3]{%
  #2 
  \begin{scope}[shift={(#1,0)}]
    #3
  \end{scope}
}

\newcolumntype{C}[1]{>{\centering\let\newline\\\arraybackslash\hspace{0pt}}m{#1}}
\theoremstyle{plain}

\usepackage{booktabs}
\usepackage{xfrac}

\newcommand{\ig}[2]{\includegraphics[scale=#1]{#2}}

\begin{document}
\renewcommand{\overset}[2]{#2^{#1}}

\maketitle


\begin{abstract}
  The paper extends Bayesian networks (BNs) by a mechanism for dynamic
  changes to the probability distributions represented by BNs.  One
  application scenario is the process of knowledge acquisition of an
  observer interacting with a system. In particular, the paper
  considers condition/event nets where the observer's knowledge about
  the current marking is a probability distribution over markings.
  The observer can interact with the net to deduce information about
  the marking by requesting certain transitions to fire and
  observing their success or failure.

  Aiming for an efficient implementation of dynamic changes to
  probability distributions of BNs, we consider a modular form of
  networks that form the arrows of a free PROP with a
  commutative comonoid structure, also known as term graphs.  The
  algebraic structure of such PROPs supplies us with a compositional
  semantics that functorially maps BNs to their underlying probability
  distribution and, in particular, it provides a convenient means to
  describe structural updates of networks.
  


\end{abstract}

\section{Introduction}
\label{sec:intro}
Representing uncertain knowledge by probability distributions %
is the core idea of Bayesian learning. %
We model the potential of an agent---%
the \emph{observer}---interacting with a concurrent system with hidden
or uncertain state to gain knowledge through ``experimenting'' with
the system, focussing on the problem of keeping track of knowledge
updates correctly and efficiently. %
Knowledge about states is represented by a probability
distribution. Our system models are condition/even nets where states
or possible worlds are markings and transitions describe which updates
are allowed. %

%



\begin{wrapfigure}{r}{0.5\textwidth}
	\centering
	\ig{0.2}{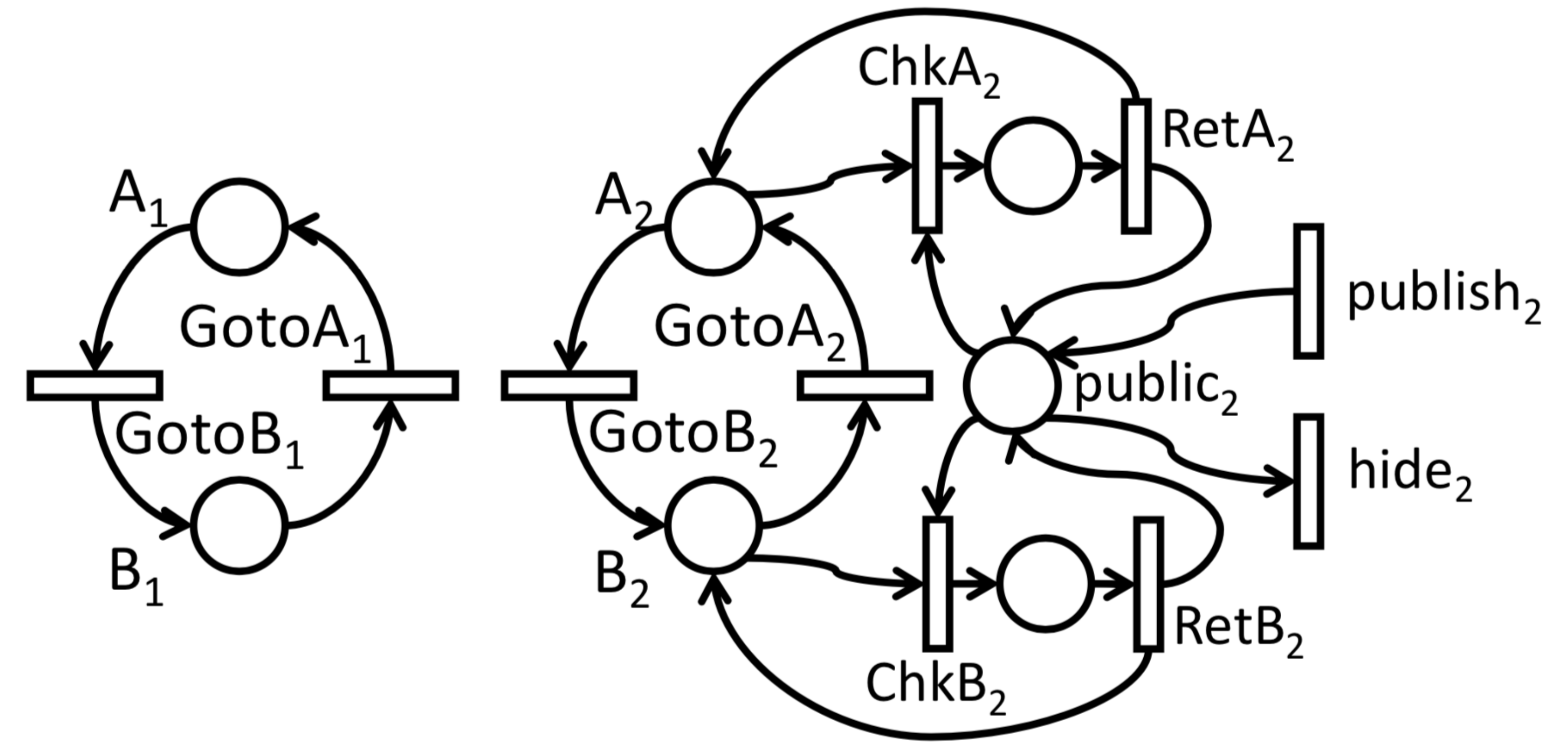}
	\vspace{-5pt}  
	\caption{Example: Social network account with location privacy}
	\label{fig:location-privacy}
\end{wrapfigure}
%
%

In order to clarify our intentions we consider an application scenario
from social media: preventing inadvertent disclosure, which is the
concern of location
privacy~\cite{DBLP:journals/geoinformatica/Damiani14}.  Consider the
example of a social network account, modelled as a condition/event
net, allowing a user to update and share their location (see
Figure~\ref{fig:location-privacy}).  We consider two users. User~1
does not allow location updates to be posted to the social network,
they are only recorded on their device. In the net this is represented
by places $\mathsf{A_1}$ and $\mathsf{B_1}$ modelling the user at
corresponding locations, and transitions $\mathsf{GotoA_1}$ and
$\mathsf{GotoB_1}$ for moving between them. We assume that only User~1
can fire or observe these transitions. User~2 has a similar structure
for locations and movements, but allows the network to track their
location. The user can decide to make their location public or hide it
by firing transition $\mathsf{publish_2}$ or $\mathsf{hide_2}$. Any
observer can attempt to fire $\mathsf{ChkA_2; RetA_2}$ or
$\mathsf{ChkB_2; RetB_2}$ to query the current location of User~2.
If~$\mathsf{public_2}$ is marked, this will allow the observer to infer
the correct location. Otherwise the observer is left uncertain as to
why the query fails, i.e. due to the wrong location being tested or
the lack of permission, unless they test both locations.  
While our net captures the totality of possible behaviours, we identify different observers, the two users, the social network, and an unrelated observer. For each of these we define which transitions they can access. We then focus on one observer and only allow  transitions they are authorised for.
In our example, if we want to analyse the unrelated observer, we fix the users' locations and privacy choices before it is the observer's turn to query the system.
 
The observer may have prior
knowledge about the dependencies between the locations of Users~1
and~2, for example due to photos with location information published by
User~2, in which both users may be identifiable. The prior knowledge
is represented in the initial probability distribution, updated
according to the observations.

We also draw inspiration from probabilistic databases
\cite{sork:prob-databases,ako:worlds-beyond} %
where the values of attributes or the presence of records are only
known probabilistically. However, an update to the database might make
it necessary to revise the probabilities. Think for instance of a
database where the gender of a person (male or female) is unknown and
we assume with probability $1/2$ that they are male. Now a record is
entered, stating that the person has married a male. Does it now
become more probable that the person is female?

Despite its simplicity, our system model based on condition/event nets
allows us to capture databases: the content of a database can
be represented as a (hyper-)graph (where each record is a
(hyper-)edge). If the nodes of the graph are fixed, updates can be
represented by the transitions of a net, where each potential record
is represented by a place.

Given a net, the observer does not know the initial marking, but has a
prior belief, given by a probability distribution over markings. The
observer can try to fire transitions and observe whether the firing is
successful or fails.  Then the probability distribution is updated
accordingly. While the update mechanism is rather straightforward, the
problem lies in the huge number of potential states: we have $2^n$
markings if $n$ is the number of places.

To mitigate this state space explosion, %
we propose to represent the observer's knowledge using Bayesian
networks (BNs) \cite{p:bayesian-networks,p:causality}, i.e., graphical
models that record conditional dependencies of random variables in a
compact form. %
However, we encounter a new problem as %
updating the observer's knowledge becomes
non-trivial. 
To do this correctly and efficiently, %
we develop a compositional approach to BNs based on symmetric monoidal
theories and PROPs
\cite{m:categorical-algebra}. 
In particular, %
we consider modular Bayesian networks as arrows of a freely generated
PROP and (sub-)stochastic matrices as another PROP with a functor from
the former to the latter. In this way, we make Bayesian networks
compositional and we obtain a graphical formalism
\cite{s:graphical-monoidal} that we use to modify Bayesian networks:
in particular, we can replace entire subgraphs of Bayesian networks by
equivalent ones, i.e., graphs that evaluate to the same matrix.  The
compositional approach allows us to specify complex updates in
Bayesian networks by a sequence of simpler updates %
using a small number of primitives. %

We furthermore describe an implementation and report promising runtime
results.

\full{The proofs of all results can be found in
  Appendix~\ref{apx:proofs}.}\short{The proofs of all results can be
  found in the full version of this paper
  \cite{chhk:update-ce-nets-bayesian-arxiv}.}

\section{Knowledge Update in Condition/Event Nets}\label{sec:general_setup}
We will formalise knowledge updates by means of an extension of Petri nets with probabilistic
knowledge on markings. %
 The starting point are condition/event nets \cite{r:petri-nets}. 


\begin{defi}[Condition/event net] \label{def:petri_net} A
  condition/event net (CN)
  $N = (S, T, \leftdot{()}, \rightdot{()}, m_0)$ is a five-tuple
  consisting of a finite set of \emph{places} $S$, a finite set of
  \emph{transitions} $T$ with \emph{pre-conditions}
  $\leftdot{()}: T \rightarrow \mathcal{P}(S)$, \emph{post-conditions}
  $\rightdot{()}: T \rightarrow \mathcal{P}(S)$, and $m_0 \subseteq S$
  an \emph{initial marking}. %
  A \emph{marking} is any subset of places \(m\subseteq S\). We
  assume that for any $t\in T$, $\leftdot t \cap \rightdot t = \emptyset$.

  A transition $t$ can \emph{fire} for a marking $m \subseteq S$,
  denoted $m \Rightarrow^t$, if $\leftdot{t} \subseteq m$ and
  $\rightdot{t}\cap m = \emptyset$. Then marking $m$ is transformed
  into $m' = (m \setminus \leftdot{t}) \cup \rightdot{t}$, written
  $m \overset{t}{\Rightarrow} m'$. We write
  $m \overset{t}{\Rightarrow}$ to indicate that there exists some $m'$
  with $m \overset{t}{\Rightarrow} m'$.

  We will use two different notations to indicate that a transition
  cannot fire, the first referring to the fact that the pre-condition
  is not sufficiently marked, the second stating that
  there are tokens in the post-condition:
  $m\overset{t}{\not\Rightarrow}_\mathit{pre}$ whenever
  $\leftdot{t} \not\subseteq m$ and
  $m\overset{t}{\not\Rightarrow}_\mathit{post}$ whenever
  $\rightdot{t} \cap m \neq \emptyset$.
  We denote the \emph{set of all markings} by
  $\mathcal{M} = \mathcal{P}(S)$.
\end{defi}

For simplicity we assume that
$S = \{1, \dots, n \}$ for $n \in \mathbb{N}$.  
Then, a marking $m$ can be characterized  by a boolean
vector $m: S \rightarrow \{0,1\}$, i.e.,
$\mathcal{M} \cong \{0,1\}^S$.  Using the vector notation
we write $m(A) = \{ 1 \}$ for $A \subseteq S$ if 
all places in $A$ are marked in $m$.

To capture the probabilistic observer we augment CNs by a
probability distribution over markings modelling 
uncertainty about the hidden initial or current marking.  

\begin{defi}[Condition/Event net with Uncertainty] 
  A \emph{Condition/Event Net with Uncertainty (CNU)} is a six-tuple
  $N = (S, T, \leftdot{()}, \rightdot{()}, m_0, p)$ where
  $(S, T, \leftdot{()}, \rightdot{()}, m_0)$ is a net as in
  Definition \ref{def:petri_net}.  Additionally, \(p\) is a function
  $p: \mathcal{M} \rightarrow [0,1]$ with
  $\sum_{m \in \mathcal{M}} p(m) = 1$ that assigns a \emph{probability
    mass} to each possible marking.  This gives rise to a probability
  space $(\mathcal{M}, \mathcal{P}(\mathcal{M}), \prob)$ with
  $\prob: \mathcal{P}(\mathcal{M}) \rightarrow [0,1]$ defined by
  $\prob\bigl(\{ m_1, \dots, m_k \}\bigr) = \sum_{i=1}^k p(m_i)$.
  
  We assume that $p(m_0) > 0$, i.e. the initial marking is possible according to $p$.
\end{defi}

We model the knowledge gained by observers when firing transitions
and observing their outcomes. Firing $t \in T$ can
either result in success (all places of $\leftdot t$ 
are marked and no place in \(\rightdot t\) is marked) 
or in failure (at least one place of $\leftdot t$ is
empty or one place in \(\rightdot t\) is marked%
). Thus, there are two kinds of failure, the absence
of tokens in the pre-condition or the presence of tokens in the
post-condition. If a transition fails for both reasons, the observer
will learn only one of them.
To model the knowledge gained we define the following operations on distributions.
\todo[inline]{
	\textbf{R1:} must none of the places of $\rightdot t$ be marked for t to be successful? It seems like the fail\_post condition implies this, but it is not stated as a success condition. \\
	\textbf{Ben:} [DONE] I think it is explicitly stated in the previous sentences and we should not change anything
}

\begin{defi}[Operations on CNUs] 
  \label{def:operations_on_CNU}
  Given a CNU $N = (S, T, \leftdot{()}, \rightdot{()}, m_0, p)$ the
  following operations update the mass function $p$ and as a result
  the probability distribution $\prob$.     
  \begin{itemize}
  \item To assert that certain places $A \subseteq S$
    \emph{all contain} a token ($b = 1$) or that none contains a
    token ($b = 0$) we define the operation \emph{assert} 
    \vspace{-1ex}
    \begin{equation}
      \label{eq:assertop}
      \assertop_{A,b}(p)(m)
      = \frac{p(m)}{\sum_{m' \in \mathcal{M}: m'(A) = \{ b \}} p(m')},
       \text{ if } m(A) = \{b\} \qquad
      \text{and}\quad
      0,\text{ otherwise. } 
    \end{equation}
  \item To state that \emph{at
      least one} place of a set $A \subseteq S$
    \emph{does} (resp.\ \emph{does not}) \emph{contain} a token 
    we define operation \emph{negative assert} 
    \vspace{-1ex}
    \begin{equation} \label{eq:nassertop} \nassertop_{A,b}(p)(m) =
      \frac{p(m)}{\sum_{m' \in
          \mathcal{M}: m'(A) \neq \{b\}} p(m')}, \text{ if } m(A) \neq
      \{b\}  \quad \text{and}\quad
      0, \text{ otherwise. }
    \end{equation}
  \item \emph{Modifying} a set of places $A \subseteq S$ such that
    \emph{all places contain a token ($b=1$)} \emph{or none contains a
      token ($b=0$)} requires the following operation
    \begin{equation}
      \label{eq:setop}
      \setop_{A,b}(p)(m) = {
        \sum_{m': m'|_{S \setminus A} = m|_{S
          \setminus A}}}
      p(m'), \text{ if } m(A) = \{b\} \quad
      \text{and}\quad 0, \text{ otherwise}.
    \end{equation}
  \item A successful firing of a transition $t$ leads to an assert
    ($\assertop$) and $\setop$ of the pre-conditions $\leftdot t$ and
    the post-conditions $\rightdot t$.  A failed firing translates to
    a negative assert ($\nassertop$) of the pre- or post-condition and
    nothing is set.  Thus we define for a transition $t \in T$
    \vspace{-1ex}
    \begin{align*}
      \successop_{t}(p) &= \setop_{\rightdot
        t,1}(\setop_{\leftdot
        t,0}(\assertop_{\rightdot t,0}(\assertop_{\leftdot t,1}(p))))
      & \failop^\mathit{pre}_{t}(p) &= 
      \nassertop_{\leftdot t,1}(p)\\
      &&\failop^\mathit{post}_{t}(p) &=
      \nassertop_{\rightdot t,0}(p).
    \end{align*}
  \end{itemize}
\end{defi}
Operations $\assertop,  \nassertop$ are partial, defined whenever the sum in the denominator of their first clause is greater than $0$. That means, the observer only fires transitions whose
pre- and postconditions have a probability greater than zero, i.e., where
according to their knowledge about the state it is possible that these transitions are enabled.
By Definition~\ref{def:petri_net} the initial marking is possible, and this property is maintained as markings and distributions are updated. If this assumption is not satisfied, the operations in Definition~\ref{def:operations_on_CNU} are undefined.

The $\assertop$ and $\nassertop$ operations result from
 conditioning the input distribution on (not) having tokens at $A$ (compare Proposition~\ref{prop:conditioned}).
Also, $\setop$ and $\assertop$ for
$A = \{s_1, \dots, s_k\} \subseteq S$ can be performed iteratively,
i.e.,
$\setop_{A,b} = \setop_{\{s_k\},b} \circ \dots \circ
\setop_{\{s_1\},b}$ and
$\assertop_{A,b} = \assertop_{\{s_k\},b} \circ \dots \circ
\assertop_{\{s_1\},b}$. For a single place $s$ we have
$\assertop_{\{s\},b} = \nassertop_{\{s\},1-b}$.

Figure~\ref{fig:update-example} gives an example for a Petri net with
uncertainty and explains how the observer can update their knowledge by
interacting with the net.
\begin{figure}
  \centering
  \noindent
  \begin{minipage}{0.29\linewidth}
    \centering
    \begin{tikzpicture}[node distance=1.3cm,>=stealth',bend angle=45,auto,scale=1.1]
      \tikzstyle{place}=[circle,thick,draw=black,minimum size=4mm]
      \tikzstyle{red place}=[place,draw=red!75,fill=red!20]
      \tikzstyle{transition}=[rectangle,thick,draw=black!75,fill=black!20,minimum size=2mm]

      \tikzstyle{every label}=[red]

      \begin{scope}[]
        \node [place] (s1) [label=above:$S_1$]{};
        \node [place] (s2) [label=below:$S_2$, below of=s1, left of=s1]                      {};
        \node [place] (s3)  [label=below:$S_3$, below of=s1, right of=s1] {};

        \node [transition] (t1) [below of=s1, label=below:$t_1$] {}
        edge [pre] (s1)
        edge [post] (s2)
        edge [post] (s3);

		\node [transition] (t2) at (-0.75,-0.75) [label=left:$t_2$] {}
        edge [pre] (s2)
        edge [post] (s1);

        \node [transition] (t3) at (0.75,-0.75)[label=above:$t_3$] {}
        edge [pre] (s3)
        edge [post] (s1);

        \node [transition] (t4) at (0,-2.5) [label=below:$t_4$] {}
        edge [pre] (s2)
        edge [post] (s3);
      \end{scope}

      \begin{scope}
      \end{scope}
    \end{tikzpicture}
  \end{minipage}
  \begin{minipage}{0.7\linewidth}\small
    \centering
    \def\arraystretch{1.1}
    \setlength\tabcolsep{3pt}
    \begin{tabularx}{\textwidth}{ccc||c|c|c|c|c|c}
      \multicolumn{3}{c||}{\textit{-- places --}} &
      \multicolumn{1}{c|}{\ }  &
      \multicolumn{4}{c|}{\textit{$\mathit{success}_{t_4}$}} & \multicolumn{1}{c}{\textit{$\mathit{fail}^\mathit{pre}_{t_1}$}} \\
      $S_1$ & $S_2$ & $S_3$ &\ init\ \ &
	  $\mathrm{as}_{\{S_2\},1}$  & $\mathrm{as}_{\{S_3\},0}$ &
	  $\setop_{\{S_2\},0}$ & $\setop_{\{S_3\},1}$ & $\mathrm{nas}_{\{S_1\},1}$ \\
      \hline
      $1$ & $1$ & $1$ & $\sfrac{1}{12}$ & $\sfrac{1}{6}$ & $0$ & $0$ & $0$ & $0$ \\
      $1$ & $1$ & $0$ & $\sfrac{1}{6}$ & $\sfrac{1}{3}$ & $\sfrac{1}{2}$ & $0$ & $0$ & $0$\\
      $1$ & $0$ & $1$ & $\sfrac{1}{8}$ & $0$ &
      $0$ & $0$ & $\sfrac{1}{2}$ & $0$\\
      $1$ & $0$ & $0$ & $\sfrac{1}{8}$ & $0$ &
      $0$ & $\sfrac{1}{2}$ & $0$ & $0$\\
      $0$ & $1$ & $1$ & $\sfrac{1}{12}$ & $\sfrac{1}{6}$ & $0$ & $0$ & $0$ & $0$\\
      $0$ & $1$ & $0$ & $\sfrac{1}{6}$ & $\sfrac{1}{3}$ & $\sfrac{1}{2}$ & $0$ & $0$ & $0$\\
      $0$ & $0$ & $1$ & $\sfrac{1}{8}$ & $0$ & $0$ & $0$ & $\sfrac{1}{2}$ & $1$\\
      $0$ & $0$ & $0$ & $\sfrac{1}{8}$ & $0$ & $0$ & $\sfrac{1}{2}$ & $0$ & $0$\\
    \end{tabularx}
  \end{minipage}
  \caption{
	  Example of operations on a net with uncertainty. 
	  We set $m_0 = \{S_2\}$ and assume the observer first fires $t_4$ (and succeeds) and then tries to fire $t_1$ (and fails). 
	  Columns in the table represent updated distributions on the markings after each operation (ordered from left to right).
	  For this example, in the end the observer knows that the final configuration is $\{S_3\}$ with probability~$1$. 
  }
  \label{fig:update-example}
\end{figure}
\todo[inline]{
	\textbf{R1:}  Figure 2: last column of the table seems wrong, should the 1 be in the row corresponding to 001 (instead of 101)? \\
	\textbf{Ben:} [DONE] He/She is right. I moved the 1 to the correct spot
}
\todo[inline]{
	\textbf{R1:} Also, it would be nice to explicitly state each column [in Fig. 2] is the result of applying the operation in the header to the column immediately to the left.\\
	\textbf{Ben:} [DONE] I changed one sentence in the caption of Fig. 2
}
We can now show that our operations correctly update the probability assumptions according to the
observations of the net.

\begin{restatable}{prop}{PropConditioned}
  \label{prop:conditioned} Let
  $N = (S, T, \leftdot{()}, \rightdot{()}, m_0, p)$ be a CNU where
  $\prob$ is the corresponding probability distribution. For given
  $t \in T$ and $m\in\mathcal{M}$ let
  $\markt{} = \{m'\in\mathcal{M}\mid m'\overset{t}{\Rightarrow}\}$,
  $\markt{m} = \{m'\in\mathcal{M}\mid m'\overset{t}{\Rightarrow} m\}$,
  $\marktnpre = \{m'\in\mathcal{M}\mid
  m'\not\overset{t}{\Rightarrow}_\mathit{pre}\}$ and
  $\marktnpost = \{m'\in\mathcal{M}\mid
  m'\not\overset{t}{\Rightarrow}_\mathit{post}\}$. Then, provided that
  $\markt{}$, $\marktnpre$ respectively $\marktnpost$ are non-empty,
  it holds for $m\in\mathcal{M}$ that
  \vspace{-1ex}
  \begin{align*}
    \successop_{t}(p)(m) 
    &= \prob(\markt{m}\mid \markt{}) 
    & \failop^\mathit{pre}_{t}(p)(m) 
    &= \prob(\{m\}\mid \marktnpre)
    \\
    &&
    \failop^\mathit{post}_{t}(p)(m) 
    &=\prob(\{m\}\mid \marktnpost)
  \end{align*}
\end{restatable}
\todo[inline]{
	\textbf{R1:} it's clear what is meant from context, but I don't think the notation $m' \Rightarrow^t$ was ever defined.\\
	\textbf{Ben:} [DONE] I added the notation $m \Rightarrow^t$ in Definition 1. Please check.
}

We shall refer to the the joint distribution (over all places)
by~$\prob$. Note that it is unfeasible to explicitly store it %
if the number of places is large. To mitigate this problem we use a
Bayesian network with a random variable for each place, %
recording dependencies between the presence of tokens in different places. %
If such dependencies are local, %
the BN is often moderate in size and %
thus provides a compact symbolic representation. %
However, updating the joint distribution of BNs is non-trivial. %
To address this problem, %
we propose a propagation procedure based on a term-based, modular
representation of BNs.

\section{%
  Modular Bayesian Networks and %
  Sub-Stochastic Matrices%
}\label{sec:MBN-substoch}



Bayesian networks (BNs) are %
a graphical formalism to reason about probability distributions. %
They are visualized as %
directed, acyclic graphs with %
nodes random variables and %
edges dependencies between them. 
\todo[inline]{
	\textbf{R1:}``the probability distribution of Figure 2'' -- Figure 2 has many probability distributions; perhaps it's clearer to refer explicitly to 'init'\\
	\textbf{Ben:} [DONE] I removed the sentence ``An example BN, encoding the probability distribution of Figure~\ref{fig:update-example} is given in Figure~\ref{fig:mbn-example}.'' completely and in Figure \ref{fig:mbn-example} I added ``initial'' to distribution.
}
This is sufficient for static BNs %
 whose most common operation is the %
inference of (marginalized or conditional) distributions %
of the underlying joint distribution.  %

For a rewriting calculus on dynamic BNs, %
we consider a modular representation of networks %
that do not only encode a single probability vector, but a matrix,
with several input and output ports.
The first aim is compositionality: %
larger nets can be composed from smaller ones via sequential and
parallel composition, %
which correspond to matrix multiplication %
and Kronecker product of the encoded matrices. %
This means, %
we can implement the operations of
Section~\ref{sec:general_setup} in a modular way.

\paragraph*{PROPs with Commutative Comonoid Structure}

We now describe the underlying compositional structure %
of (modular) BNs and \mbox{(sub-)}\allowbreak stochastic matrices, %
which facilitates a compositional computation of %
the underlying probability distribution of %
(modular) BNs. %
The mathematical structure are
PROPs
~\cite{m:categorical-algebra} %
(see also \cite[Chapter~5.2]{2018arXiv180305316F}), %
i.e., %
strict symmetric monoidal categories %
\((C,\otimes, 0, \sigma)\) %
whose objects are %
(in bijection with) %
the natural numbers, with monoidal product \(\otimes\) as %
(essentially) addition, with unit $0$. %
The compositional structure of PROPs can be intuitively %
represented using string diagrams with wires and boxes %
(see Figure~\ref{fig:prop-composition}). Symmetries \(\sigma\) serve
for the reordering of wires. \todo{B: What is $0$?}%
\begin{figure}[htb]
  \centering
  \tikzstyle{every picture}+=[scale=1.5,baseline={([yshift=-1ex]current
    bounding box.west)}]
  \begin{tabular}[c]{cc}
      \(
      \tikz[xscale=.7]\Rcell[1.7]{m}{k}{f;f'};
      =                          
      \tikz[xscale=.7]\cmp{1}{\Rcell{m}{n}{f}}{\Rcell{n}{k}{f'}}; 
    \)                             
                  &
                        \(
  \tikz[xscale=.8]\Rcell[2.5]{m_1+m_2}{n_1+n_2}{f_1\otimes f_2};
                    =
  \begin{array}[c]{c}
        \tikz[xscale=.9]\Rcell{m_1}{n_1}{f_1};
    ~\\[2ex]
    \tikz[xscale=.9]\Rcell{m_{2}}{n_{2}}{f_{2}};
  \end{array}\)
  \end{tabular}
  \caption{%
    String diagrammatic composition (resp.\ tensor) of two arrows
    \(f\colon m \to n\), \(f'\colon n \to k\)  
    (resp.\ \(f_1 \colon m_1 \to n_1\), \(f_2 \colon m_2 \to n_2\)) %
    of a PROP \((C,\otimes,0,\sigma)\)  %
  }
  \label{fig:prop-composition}
\end{figure}

A paradigmatic example is %
the PROP of \(2^n\)-dimensional Euclidean spaces %
and linear maps, %
equipped with the tensor product: %
the tensor product of \(2^n\)- and \(2^m\)-dimensional spaces is
\(2^{n+m}\)-dimensional, %
composition of linear maps amounts to matrix multiplication, %
and the tensor product is also known as Kronecker product %
(as detailed below). %
We refer to the natural numbers of the domain and codomain of %
arrows in a PROP as their \emph{type}; %
thus, %
a linear map from \(2^n\)- to \(2^m\)-dimensional Euclidean space %
has type \(n\to m\). %

We shall have the %
additional structure on symmetric monoidal categories %
that was dubbed \emph{graph substitution} %
in work on term graphs~\cite{cg:term-graphs-gs-monoidal}, %
which amounts to a commutative comonoid structure on PROPs.


\begin{defi}[PROPs with commutative comonoid structure]
  A \emph{CC-structured PROP} %
  is a tuple \((C,\otimes, 0, \sigma,\nabla,\top)\) %
  where \((C,\otimes, 0, \sigma)\) is a PROP %
  and the last two components are
  arrows \(\nabla\colon 1 \to 2\) and \(\top\colon 1\to 0\), 
  which are subject to Equations~\ref{eq:cc-structure} 
  (cf.\  Figure~\ref{fig:comonoid-gates}).
  \begin{equation}
    \label{eq:cc-structure}
        \nabla;(\nabla\otimes \id_1) = \nabla;(\id_1\otimes\nabla), 
    \nabla = \nabla;\sigma\qquad\qquad
    \nabla;(\id_1\otimes \top) = \id_1
  \end{equation}
\end{defi}

\begin{table}
  \centering
  \cbox{
  \begin{eqnarray*}
    && (t_1;t_3) \otimes (t_2;t_4) = (t_1\otimes t_2);(t_3\otimes
    t_4) \qquad
    (t_1;t_2);t_3 = t_1;(t_2;t_3) \\
    && \id_n;t = t = t;\id_m \qquad
    (t_1\otimes t_2)\otimes t_3 = t_1\otimes(t_2\otimes t_3) \qquad
    \id_0 \otimes t = t = t\otimes \id_0 \\
    && \sigma;\sigma = \id_2 \qquad
    (t\otimes\id_m);\sigma_{n,m} = \sigma_{m,n};(\id_n\otimes t)
    \qquad
    \nabla;(\nabla\otimes \id_1) = \nabla;(\id_1\otimes\nabla)
    \\
    && \nabla = \nabla;\sigma  \qquad 
    \nabla;(\id_1\otimes \top) = \id_1 
  \end{eqnarray*}
  \hrule
  \begin{eqnarray*}
    && \id_1 = \id \qquad \id_{n+1} = \id_n\otimes \id_1 \\
    && \sigma_{n,0} = \sigma_{0,n} = \id_n \qquad 
    \sigma_{n+1,1} = (\id\otimes \sigma_{n,1});(\sigma\otimes \id_n)
    \\
    && \qquad \sigma_{n,m+1} = (\sigma_{n,m}\otimes \id_1);(\id_m\otimes
    \sigma_{n,1}) \\
    && \nabla_1 = \nabla \qquad \nabla_{n+1} = (\nabla_n\otimes
    \nabla);(\id_n \otimes \sigma_{n,1}\otimes \id) \\
    && \top_1 = \top\qquad \top_{n+1} = \top_n\otimes \top
  \end{eqnarray*}
  }
  \caption{Axioms for CC-structured PROPs and definition of operators
    of higher arity}
  \label{tab:axioms-cc-prop}
\end{table}

\begin{wrapfigure}{r}{0.52\textwidth}
    \centering
    \includegraphics[width=0.5\textwidth]{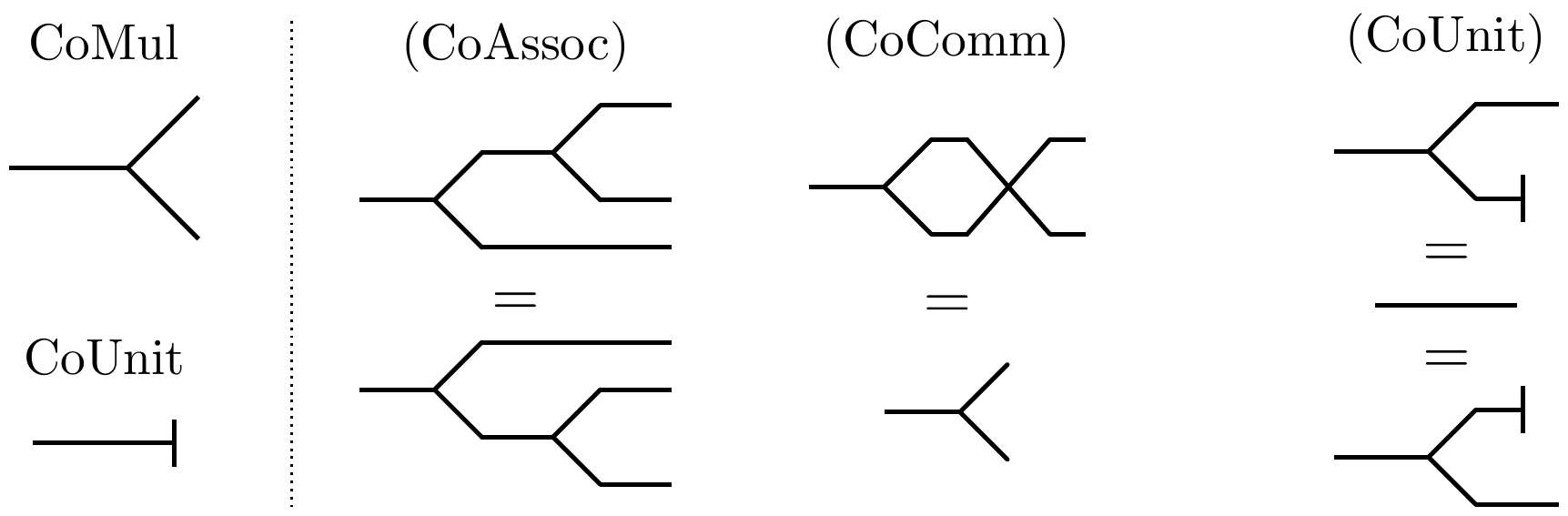}
    \caption{Comultiplication and counit
      arrows\\ 
      and the equations of commutative comonoids %
    }
    \label{fig:comonoid-gates}
\end{wrapfigure}
To give another, more direct definition, the arrows of a freely
generated CC-structured PROP can be represented as terms over some set
of generators $g\in G$ and constants $\id\colon 1\to 1$,
$\sigma\colon 2\to 2$, $\nabla\colon 1\to 2$, $\top\colon 1\to 0$,
combined with the operators sequential composition ($;$) and tensor
($\otimes$) and quotiented by the axioms in
Table~\ref{tab:axioms-cc-prop} (see \cite{z:interacting-hopf}). This
table also lists the definition of operators of higher arity. We often
refer to the comultiplication~\(\Delta\) and %
its counit~\(\top\) as \emph{duplicator} and \emph{terminator},
resp.  %
(cf.\ Figure~\ref{fig:comonoid-gates}). %
Roughly, adding the commutative comonoid structure %
amounts to the possibility to have several or no connections to each one
of the output port of gates %
and input ports. %
In other words, %
outputs can be shared.

\paragraph*{(Sub-)Stochastic Matrices}
We now consider (sub-)stochastic matrices as an instance of a
CC-structured PROP.  A \emph{matrix} of type $n\to m$ is a matrix $P$
of dimension $2^m\times 2^n$ with entries taken from the closed
interval $[0,1] \subseteq \mathbb{R}$. We restrict attention to
sub-stochastic matrices, i.e., column sums will be at most~$1$; if we
require equality, we obtain stochastic matrices.

\begin{wrapfigure}{r}{0.282\textwidth}
  \begin{center}
    ~\\[-6ex]
    $
    \begin{matrix}
      11 \\ 10 \\ 01 \\ 00
    \end{matrix}
    \begin{pmatrix}
      \    \smash{\raisebox{3ex}{\makebox[0pt]{11}}}
      {\cdot}\  &
      \ {\cdot}\smash{\raisebox{3ex}{\makebox[0pt]{10}}}
      \  & \ {\cdot}\smash{\raisebox{3ex}{\makebox[0pt]{01}}}
      \  & \ {\cdot}\smash{\raisebox{3ex}{\makebox[0pt]{00}}}\ \ \  \\
      \ \cdot\  & \ \cdot\  & \ \cdot\  & \ \cdot\ \ \\
      \ \cdot\  & \ \cdot\  & \ \cdot\  & \ \cdot\ \ \\
      \ \cdot\ & \ \cdot\ & \ \cdot\ & \ \cdot\ \
    \end{pmatrix}$\vspace{-2ex}
  \end{center}
\end{wrapfigure}
We index matrices over $\{0,1\}^m \times \{0,1\}^n$, i.e., for
$\bitvec{x}\in \{0,1\}^m$, ${\bitvec{y}\in \{0,1\}^n}$ the
corresponding entry is denoted by $P(\bitvec{x}\mid \bitvec{y})$. We
use this notation to evoke the idea of conditional probability (the
probability that the first index is equal to $\bitvec{x}$, whenever
the second index is equal to $\bitvec{y}$.) When we write $P$ as a
matrix, the rows/columns are ordered according to a descending
sequence of binary numbers ($1\dots 1$ first, $0\dots 0$ last).

Sequential composition is \emph{matrix multiplication}, i.e.,
given $P\colon n\to m$, $Q\colon m\to \ell$ we define
$P;Q = Q\cdot P\colon n\to \ell$, which is a
$2^\ell\times 2^n$-matrix. The tensor is given by the \emph{Kronecker
product}, i.e., given $P\colon n_1\to m_1$, $Q\colon n_2\to m_2$ we
define $P\otimes Q\colon n_1+n_2\to m_1+m_2$ as
$(P\otimes Q)(\bitvec{x}_1\bitvec{x}_2\mid \bitvec{y}_1\bitvec{y}_2) =
P(\bitvec{x}_1\mid \bitvec{y}_1)\cdot Q(\bitvec{x}_2\mid \bitvec{y}_2)$ where
$\bitvec{x}_i \in \{0,1\}^{n_i}$, $\bitvec{y}_i \in \{0,1\}^{m_i}$.

The constants are defined as follows:
\begin{eqnarray*}
  \id_0 = (1) \quad \id =
  \begin{pmatrix}
    1 & 0 \\ 0 & 1
  \end{pmatrix} \quad \nabla =
  \begin{pmatrix}
    1 & 0 \\ 0 & 0 \\ 0 & 0 \\ 0 & 1
  \end{pmatrix} \quad \sigma =
  \begin{pmatrix}
    1 & 0 & 0 & 0 \\ 0 & 0 & 1 & 0 \\ 0 & 1 & 0 & 0 \\ 0 & 0 & 0 & 1
  \end{pmatrix} \quad \top =
  \begin{pmatrix}
    1 & 1
  \end{pmatrix} 
\end{eqnarray*}

\full{In more detail, the constant matrices can be spelled out as
  follows.

  \begin{itemize}
  \item $\id_0$ is the unique stochastic $1\times 1$-matrix, i.e.,
    $\id_0(\epsilon,\epsilon) = 1$.
  \item $\id$ is the $2\times 2$ identity matrix, i.e., $\id(x,y) = 1$
    iff $x=y$ (otherwise $0$). 
  \item $\nabla(x_1 x_2, y) = 1$ iff $x_1 = x_2 = y$ (otherwise $0$).
  \item $\sigma(x_1 x_2, y_1 y_2) = 1$ iff $x_1 = y_1$ and $x_2 = y_1$
    (otherwise $0$).
  \item $\top(\epsilon,x) = 1$ for every $x$.
  \end{itemize}
}

\begin{prop}[\protect\cite{f:causal-theories}]
  \label{prop:stoch-mat-cc-prop}
  (Sub-)stochastic matrices form a CC-structured PROP. %
\end{prop}
\full{
  \begin{proof}[Proof sketch]
    It is straightforward to check that (sub-)stochastic matrices
    satisfy all the axioms in Table~\ref{tab:axioms-cc-prop}. On
    the other hand the result also follows from
    \cite{f:causal-theories}, which interprets Bayesian networks over
    stochastic maps, a generalization of stochastic matrices in terms of
    measure theory.
  \end{proof}
}

\paragraph*{Causality Graphs}
We next introduce causality graphs, %
a variant of term graphs~\cite{cg:term-graphs-gs-monoidal}, to provide a modular representation of Bayesian networks.
%
Nodes play the role of gates of string diagrams; %
the main difference to port graphs \cite[Chapter 5]{2018arXiv180305316F}
is the branching structure at output ports,  %
which corresponds to (freely) added  commutative comonoid structure. 
We fix a set of generators~$G$ (a.k.a.\ signature), %
elements of which can be thought of as %
blueprints of gates of a certain type; %
all generators \(g\in G\) will be of type \(n\to 1\), %
which means that each node can be identified with its single output
port %
while it has a certain number of input ports. %

\begin{defi}[Causality Graph (CG)]
  A \emph{causality graph (CG)} of type $n\to m$ is %
  a tuple $B = (V,\ell,s,\out)$ where
  \begin{itemize}
  \item $V$ is a set of \emph{nodes},
  \item $\ell\colon V\to G$ is a \emph{labelling function} that assigns
    a generator $\ell(v)\in G$ to each node \(v \in V\), %
  \item $s\colon V\to W_B^*$ where $W_B = V\cup\{i_1,\dots,i_n\}$ %
    is the \emph{source function} that
    assigns a sequence of \emph{wires}
    \(s(v)\) to each node \(v\in V\) %
    such that %
    $|s(v)| = n$ if $\ell(v)\colon n\to 1$, %
  \item %
    $\out\colon \{o_1,\dots,o_m\}\to W_B$ is the
	\emph{output function} that assigns each output port to a wire.
  \end{itemize}
  Moreover, the corresponding directed graph (defined by $s$) has to be acyclic.

\end{defi}
\todo[inline]{
	\textbf{R1:} ``such that the corresponding DAG of the causality graph is acyclic''. This could be improved. The word "the" is repeated twice. DAGs are by definition acyclic. Perhaps most importantly, the "corresponding DAG" of a causality graph is not defined. I can guess what it means, but it would be nice to see something precise to check my intuition against. \\
	\textbf{Ben:} [MAYBE DONE] I changed the last sentence a bit. However, to properly define the directed graph we mean would need a lot more space, which would be a bit of a waste I think.
}

By $\{i_1,\dots,i_n\}$ we denote the set of \emph{input ports} and by
$\{o_1,\dots,o_m\}$ the set of \emph{output ports}.
By $\mathrm{pred}$ and $\mathrm{succ}$ we denote the direct
predecessors and successors of a node, i.e.
$\mathrm{pred}(v_0) = \{ v \in V \mid v\in s(v_0)\}$ and
$\mathrm{succ}(v_0) = \{ v \in V \mid v_0\in s(v)\}$, %
respectively. %
By $\mathrm{pred}^*(v_0)$ we denote the set of indirect predecessors,
using transitive closure. Furthermore $\mathrm{path}(v,w)$ denotes the
set of all nodes which lie on paths from $v$ to $w$.

A wire originates from a single input port or node and each node can
feed into several successor nodes and/or output ports.  Note that 
input and output are not symmetric in %
the context of causality graphs. %
This is a consequence of the absence of a monoid structure.

We equip CGs %
with operations of composition and tensor product, %
identities, %
and a commutative comonoid structure. %
We require that the node sets of Bayesian nets $B_1,B_2$ are
disjoint.%
\footnote{%
  The case of non-disjoint sets can be handled by %
  a suitable choice of coproducts. %
} %


\begin{description}
  
\item[Composition] Whenever $m_1 = n_2$, we define
  $B_1;B_2 := B = (V,\ell,s,\mathit{out})\colon n_1\to m_2$ with
  $V = V_1\uplus V_2$, $\ell = \ell_1\uplus \ell_2$,
  $s = s_1\uplus c\circ s_2$, $\mathit{out} = c\circ\mathit{out}_2$
  where $c\colon W_{B_2}\to W_B$ is defined as follows and extended to
  sequences: $c(w) = w$ if $w\in V_2$ and $c(w) = \mathit{out}_1(o_j)$
  if $w = i_j$.
\item [Tensor] Disjoint union is parallel composition, i.e.,
  $B_1\otimes B_2 := B = (V,\ell,s,\mathit{out})\colon n_1+n_2\to
  m_1+m_2$ with $V = V_1\uplus V_2$, $\ell = \ell_1\uplus \ell_2$,
  $s = s_1\uplus d\circ s_2$, where $d\colon W_{B_2}\to W_B$ and
  $\mathit{out}\colon \{o_1,\dots,o_{m_1+m_2}\}\to W_B$ are defined as
  follows: $d(w) = w$ if $w\in V_2$ and $d(w) = i_{n_1+j}$ if $w =
  i_j$. Furthermore $\mathit{out}(o_j) = \mathit{out}_1(o_j)$ if $1\le
  j\le m_1$ and $\mathit{out}(o_j) = \mathit{out}_2(o_{j-m_1})$ if
  $m_1 < j \le m_1+m_2$. 
\item[Operators] Finally the constants and generators are as
  follows:\footnote{A function $f\colon A\to B$, where
    $A = \{a_1,\dots,a_k\}$ is finite, is denoted by
    $f = [a_1\mapsto f(a_1),\dots, a_k\mapsto f(a_k)]$. We denote a
    function with empty domain by $[\,]$.}
  \begin{center}
    $\id_0= (\emptyset,[\,],[\,],[\,])\colon 0\to 0$ \quad
    $\id= (\emptyset,[\,],[\,],[o_1\mapsto i_1])\colon 1\to 1$ \quad
    $\top = (\emptyset,[\,],[\,],[\,])\colon 1\to 0$ \\
    $\sigma = (\emptyset,[\,],[\,],[o_1\mapsto i_2, o_2\mapsto i_1])\colon
    2\to 2$ \quad
    $\nabla = (\emptyset,[\,],[\,],[o_1\mapsto i_1, o_2\mapsto i_1])\colon
    1\to 2$  \\
    $B_g = (\{v\},[v\mapsto g],[v\mapsto i_1\dots i_n], [o_1\mapsto
    v])\colon n\to 1$, whenever $g\in G$ with type $g\colon n\to 1$
  \end{center}
\end{description}

Finally, %
all these operations %
lift to isomorphism classes of CGs. 

\begin{prop}[\protect\cite{cg:term-graphs-gs-monoidal}]
  \label{prop:free-cc-prop}
  CGs quotiented by isomorphism form the freely generated
  CC-structured PROP over the set of generators $G$, %
  where
  two causality graphs
  $B_i = (V_i,\ell_i,s_i,\mathit{out}_i)\colon \allowbreak n\to m$,
  $i\in\{1,2\}$, are isomorphic if there is %
  a bijective mapping $\phi\colon V_1\to V_2$ such that %
  $\ell_1(v) = \ell_2(\phi(v))$ and $\phi(s_1(v)) = s_2(v)$ hold for
  all $v\in V_1$ and $\phi(\mathit{out}_1(o_i)) = \mathit{out}_2(o_i)$
  holds %
  for all $i \in \{1,\dotsc, m\}$.\footnote{%
    We apply $\phi$ to a sequence of wires, %
    by applying $\phi$ pointwise and %
    assuming that $\phi(i_j) = i_j$ for $1\le j\le n$.%
  }
\end{prop}
\full{%
	\begin{proof}[Proof sketch]
		This follows from the fact that CC-structured PROPs correspond to
		the gs-monoidal categories (with natural numbers as objects) of
		\cite{cg:term-graphs-gs-monoidal}.  Furthermore CGs are in essence
		term graphs, where the input ports are called empty nodes. Since
		\cite{cg:term-graphs-gs-monoidal} shows that term graphs are
		one-to-one with the arrows of the free gs-monoidal category, our
		result follows.
	\end{proof}%
      }%


In the following, we often decompose a CG into a subgraph and its ``context''. 
\begin{restatable}[Decompositionality of CGs]{lem}{LemMbnSplitting}
  \label{lem:mbn_splitting} 
  Let $B = (V,\ell,s,\out)\colon n\to m$ be a causality graph.  Let
  $V' \subseteq V$ be a subset of nodes closed with respect to paths,
  i.e. for all $v,w \in V': \mathrm{path}(v,w) \subseteq V'$.  Then
  there exist $k \in \mathbb{N}$ and $(B_i, e_i)$ with
  $B_i = (V_i, l_i, s_i, \mathrm{out}_i)$ for $i=1,\dots,3$ such that
  $V_2 = V'$, $B = B_1; (\id_{k} \otimes B_2); B_3$ and
  $\out_2(o_i) \in V'$ for all $i$.
\end{restatable}
Thus, %
given a set of nodes in a BN that contains %
all nodes on paths between them, %
we have the induced subnet of the node set %
and a suitable ``context'' %
such that the whole net can be seen as the result of  %
substition of the subnet into the ``context''.

\paragraph*{Modular Bayesian Networks} 

We will now equip the nodes of causality graphs with matrices,
assigning an interpretation to each generator. %
This fully determines the corresponding matrix of the BN. Note that
Bayesian networks as PROPs have earlier been studied in
\cite{f:causal-theories,jz:transformer-bayesian,jz:influence-bayesian}.

\begin{defi}[Modular Bayesian network (MBN)]
  A \emph{modular Bayesian network (MBN)} is a tuple $(B,e)$ where
  $B = (V,\ell,s,\out)$ is a causality graph and $\ev$ an
  \emph{evaluation function} that assigns to every generator
  $g\in G$ with $g\colon n\to 1$ a $2^n\times 2$-matrix $e(g)$.  An MBN
  $(B,e)$ is called an \emph{ordinary Bayesian network (OBN)} whenever $B$ has no inputs (i.e. $B: 0 \to m$), $\out$ is a
  bijection, and every node is associated with a stochastic matrix.
\end{defi}
\todo[inline]{
	\textbf{R1:} ``B is of type 0 -> m, i.e. B has no inputs, out is a bijection...'' -- the scope of the "i.e." is not clear; one might read it as saying that being of type 0 -> m implies that B has no inputs AND that out is a bijection AND that every node is associated with a stochastic matrix.\\
	\textbf{Ben:} [DONE] Changed word order in the sentence. Now it should be clearer.
}
\todo[inline]{
	\textbf{R3:} move definition 10 a bit earlier : I found the Causality graphs subsection very dense and hardly motivated, also the substochatic matrix section seems to add very little. I think if you turn things around (and potentially take some of this material out) it would read much better.\\
	\textbf{Ben:} Easier said than done. Def. 10 depends on all previous... Maybe we can leave some things out, but I don't really see anything really unnecessary.
}

\begin{wrapfigure}{l}{0.26\textwidth}
  \centering
  \vspace{-1ex}
  \includegraphics[width=0.18\textwidth]{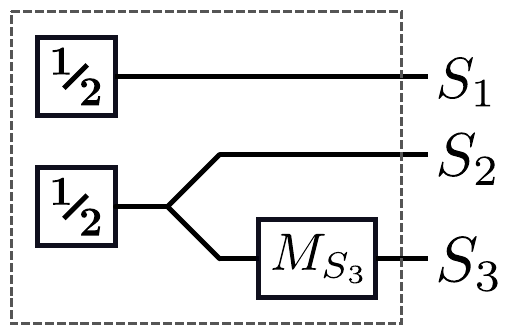}
  \caption{The initial distribution of the CNU from
    Figure~\ref{fig:update-example} as an MBN.}
  \label{fig:mbn-example}
\end{wrapfigure}

In an OBN every node $V$ corresponds to a random variable and it
represents a probability distribution on $\{0,1\}^m$.  OBNs are
exactly the Bayesian networks considered in
\cite{fgg:bayesian-network-classifiers}.\\[2ex]
\begin{minipage}[c]{1.0\linewidth}
\begin{ex}
  Figure~\ref{fig:mbn-example} gives an example of a BN where
  $\sfrac{\mathbf{1}}{\mathbf{2}} = \begin{pmatrix}\sfrac{1}{2}\\
    \sfrac{1}{2}\end{pmatrix}$ and $M_{S_3} = \begin{pmatrix}
    \sfrac{1}{3} & \sfrac{1}{2} \\
    \sfrac{2}{3} & \sfrac{1}{2}
  \end{pmatrix}$. It encodes exactly the probability distribution from
  Figure~\ref{fig:update-example}. Its term representation is
  $(g_1 \otimes (g_2;\nabla));(\id_2 \otimes g_3)$ where
  $e(g_1) = e(g_2) = \sfrac{\mathbf{1}}{\mathbf{2}}$ and
  $e(g_3) = M_{S_3}$.
\end{ex}
\end{minipage}
\medskip



\begin{defi}[MBN semantics]
  Let \((B,e)\) be an MBN where %
  the network $B = (V,\ell,s,\mathit{out})$ is of type
$n\to m$. 
The \emph{MBN semantics} is the matrix $M_e(B)$ with 
\[ \Bigl(M_e(B)\Bigr)(x_1\dots x_m\mid y_1\dots y_n) = \sum_{\substack{b\colon
      W_B\to \{0,1\}\\ b(i_j) = y_j,
      b(\textit{out}(o_i)) = x_i}} \quad \prod_{v\in V} \quad
  e\bigl(\ell(v)\bigr)\Bigl(b(v)\Bigm| b(s(v))\Bigr) \] with
$x_1,\dots,x_m,y_1,\dots,y_m\in\{0,1\}$  where $b$ is
applied pointwise to sequences.
\end{defi}

Intuitively the function $b$ assigns boolean values to wires, in a way
that is consistent with the input/output values ($x_1\dots x_m, y_1\dots y_n$). For each such assignment, the
corresponding entries in the matrices $\ell(v)$ are
multiplied. Finally, we sum over all possible assignments.

\begin{remark}
  The semantics $M_e(B)$ is compositional. %
  It is the canonical (i.e., free) extension of %
  the evaluation map %
  from single nodes to the causality graph of an MBN~\((B,e)\). %
  Here, we rely on two different findings from the literature, %
  namely, %
  the CC-PROP structure of (sub-)stochastic
  matrices~\cite{f:causal-theories} and the characterization of term
  graphs as the %
  free symmetric monoidal category with graph substition
  \cite{cg:term-graphs-gs-monoidal}. %
  \short{The formal details can be found in
    \cite{chhk:update-ce-nets-bayesian-arxiv}.}\full{The formal
    details can be found in the appendix, see
    Lemma~\ref{lem:mapping-m}.} %
\end{remark}




\section{Updating Bayesian Networks}

We have introduced MBNs as a compact and compositional representation
of distributions on markings of a CNU.  Coming back to the scenario of
knowledge update, %
we now describe how success and failure of operations requested by the
observer affect the MBN.  We will first describe how the operations
can be formulated as matrix operations that tell us which nodes have
to be added to the MBN.  We shall see that updated MBNs are in general
not OBNs, which makes it harder to interpret and retrieve the encoded
distribution.  However, we shall show that MBNs can efficiently be
reduced to OBNs.

\smallskip

\noindent\textbf{Notation:} In this section we will use the following
notation: first, we will use variants
$\id_n,\nabla_n,\sigma_{n,m},\top_n$ of the operators/matrices
$\id,\nabla,\sigma,\top$, which have a higher arity (see the
definitions in Table~\ref{tab:axioms-cc-prop}).  Furthermore, we will
write $\prod_{i=1}^k P_i$ for $P_1\cdot \ldots \cdot P_k$ and
$\bigotimes_{i=1}^k P_i$ for $P_1 \otimes \cdots \otimes P_k$. By
$0: 1 \rightarrow 1$ we denote the $2 \times 2$ zero matrix and set
$0_k = \bigotimes_{i=1}^k 0$.  We also introduce $\mathbf{1}_b$ as a
notation for the matrix~$\begin{pmatrix}1\\0\end{pmatrix}$ if $b=1$
(respectively $\begin{pmatrix}0\\1\end{pmatrix}$ if $b=0$).

With $\mathrm{diag}(a_1,\dots,a_n)$ we denote a square matrix with
entries $a_1,\dots,a_n\in[0,1]$ on the diagonal and zero elsewhere.
In particular, we will need the sub-stochastic matrices
$F_{k,b}: k \to k$ where
$F_{k,0} = \mathrm{diag}(\underbrace{1, \dots, 1}_{2^k -1 \text{ times
  }},0)$ and
$F_{k,1} = \mathrm{diag}(0,\underbrace{1, \dots, 1}_{2^k -1 \text{
    times }})$.

Given a bit-vector $\bitvec x \in \{0,1\}^n$, we will write
$\bitvec{x}_{[i]}$ respectively $\bitvec{x}_{[i\dots j]}$ to denote the
$i$-th entry respectively the sub-sequence from position $i$ to
position $j$. If $A\subseteq \{1,\dots,n\}$ we define
$\bitvec{x}_{[A]} = \{\bitvec{x}_{[i]}\mid i\in A\}$.

\paragraph*{CNU Operations on MBNs}

In this section we characterize the operations of
Definition~\ref{def:operations_on_CNU} as stochastic matrices that can
be multiplied with the distribution to perform the update.  We describe
them as compositions of smaller matrices that can easily be
interpreted as changes to an MBN.  In the following lemmas, 
$P: 0 \rightarrow m$ is always  a stochastic matrix representing the
distribution of markings of a CNU.  Furthermore,  $A \subseteq S$ is
a set of places and w.l.o.g.\ we assume that
$A = \{1,\dots,k\}$ for some $k \le m$ (as otherwise we can use
permutations that preceed and follow the operations and switch
wires as needed).

Starting with the $\setop_{A,b}$ operation \eqref{eq:setop} recall
that it is defined in a way so that the marginal distributions of
non-affected places $S \backslash A$ stay the same while the marginals
of every single place in $A$ report $b \in \{0,1\}$ with probability
one.  The following lemma shows how the matrix for a set operation can
be constructed (see Figure~\ref{fig:operations_unpropagated}).

\begin{restatable}{lem}{LemSetMatrix}
  \label{lem:set_matrix}
  It holds that
  $\setop_{A,b}(P) = \left(\bigotimes_{i=1}^m T_{A,b}^\setop(i)\right)
  \cdot P$ where $T_{A,b}^\setop(i)$ is $\mathbf{1}_b \cdot \top$ if
  $i \in A$, and $\id$ otherwise.  Moreover,
  $\bigotimes_{i=1}^m T_{A,b}^\setop(i)$ is stochastic.
\end{restatable}

\begin{figure}
  \center
  \includegraphics[width=0.7\textwidth]{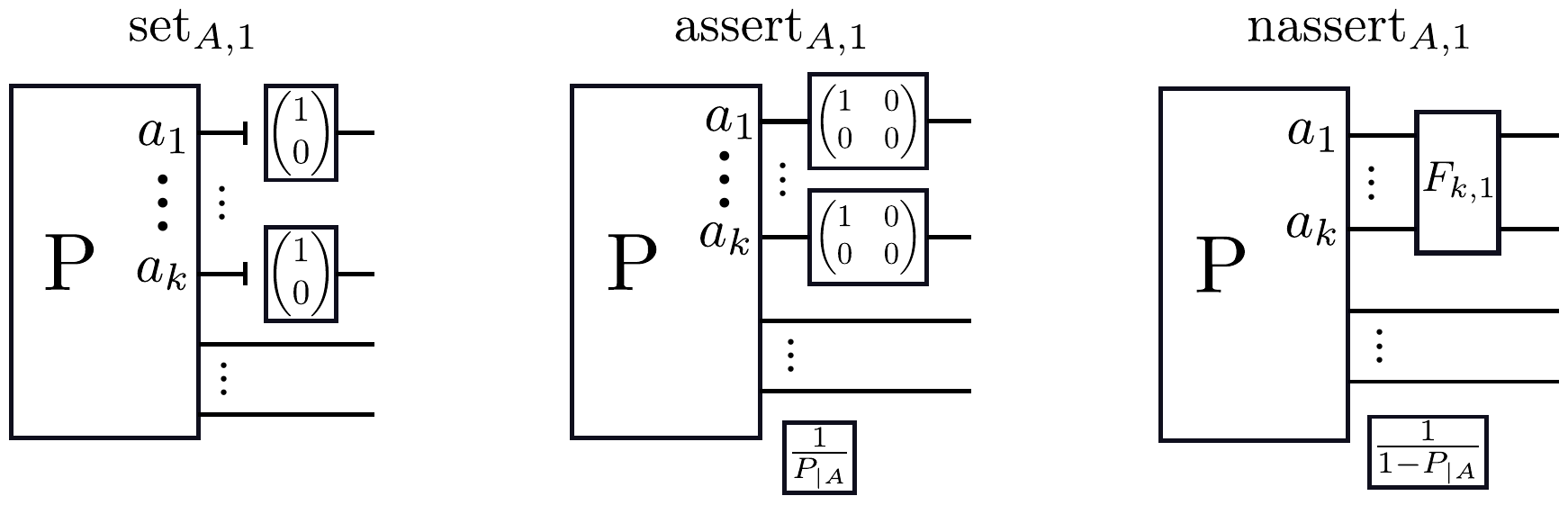}
  \caption{String diagrams of the updated distributions after
    $\mathrm{set}_{A,1}$, $\assertop_{A,1}$,
    $\nassertop_{A,1}$ operations were applied to an initial
    distribution $P$.}
  \label{fig:operations_unpropagated}
\end{figure}

Next, we deal with the $\assertop$ operation.  Applying it to a
distribution $P$ is simply a conditioning of $P$ on non-emptiness of
all places $A$.  Intuitively, this means that we keep only entries of
$P$ for which the condition is satisfied and set all other entries to
zero.  However, in order to keep the updated $P$ a probability
distribution, we have to renormalize, which already shows that
modelling this operation introduces sub-stochastic matrices to the
computation. In the next lemma normalization involves the costly
computation of a marginal $P_{|A}$ (the probability that all places in
$A$ are set to $b$), however omitting the normalization factor will
give us a sub-stochastic matrix and we will later show how such
sub-stochastic matrices can be removed, in many cases avoiding the
full costs of a marginal computation.


\begin{restatable}{lem}{LemAssertMatrix} 
  \label{lem:assert_matrix}
  It holds that
  $\assertop_{A,b}(P) = \frac{1}{P_{|A}}\left(\bigotimes_{i=1}^m
    T_{A,b}^\assertop(i)\right) \cdot P$ with
  $P_{|A} = \left(\bigotimes_{i=1}^m Q_A(i)\right) \cdot P$ where
  $T_{A,b}^\assertop(i)$ is $F_{1,1-b}$ if $i\in A$, and $\id$
  otherwise. We require that $P_{|A}\neq 0$. Furthermore
  $ Q_{A,b}^\assertop(i) =
  \begin{pmatrix}1&0\end{pmatrix}$ if $i \in A$ and $\top$ otherwise.

\end{restatable}

In contrast to $\setop$ and $\assertop$, the $\nassertop$ operation
couples all involved places in $A$.  Asserting that at least one place
has no token means that once the observer learns that e.g. one
particular place definitely has a token it affects all the other ones.
Thus for updating the distribution we have to pass the wires of places
$A$ through another matrix that removes the possibility of all places
containing a token and renormalizes.

\begin{restatable}{lem}{LemNassertMatrix}
  \label{lem:nassert_matrix}
  The following characterization holds:
  $\nassertop_{A,1}(P) = \frac{1}{P_{|A}^c}\left(F_{k,1} \otimes
    \id_{m-k}\right) \cdot P$ with $P_{|A}^c = 1 - P_{|A}$ ($P_{|A}$
  is defined as in Lemma~\ref{lem:assert_matrix}). We require that
  $P_{|A}^c\neq 0$.
\end{restatable}
\todo{
	\textbf{R3:} text after Lemma 15 is hard to read, please structure (and format) better\\
	\textbf{Ben:} I don't see how to do that without giving symbols to the 2x2 matrices...
}
An analogous result holds for $\nassertop_{A,0}$ by using $F_{k,0}$.

The previous lemmas determine how to update an MBN
$(B,e)$ to incorporate the changes to the encoded distribution
stemming from the operations on the CNU.  We denote the updated MBN by
$(B',e')$ with $B'=(V',\ell',s',\out')$.

For the $\setop_{A,b}$ operation Lemma~\ref{lem:set_matrix} shows that
we have to add a new node $v_s$ and a new generator $g_s$ for each
$s \in A$.  We set $\ell(v_s) = g_s$ and
$e'(g_s) = \mathbf{1}_b \cdot \top
= \begin{pmatrix}1&1\\0&0\end{pmatrix}$, $s'(v_s) = \out(o_s)$ and
$\out'(o_s) = v_s$.  Similarly, this holds for the $\assertop$
operation with the only difference that the associated matrix for each
$v_s$ is $\begin{pmatrix}1 & 0\\0&0\end{pmatrix}$ 
(cf. Figure~\ref{fig:operations_unpropagated}).

For the $\nassertop_{A,b}$ operation Lemma~\ref{lem:nassert_matrix}
defines a usually larger matrix $F_{k,b}: k \to k$ that intuitively
couples the random variables for all places in $A$.  We cannot simply
add a node to the MBN which evaluates to $F_{k,b}$ since nodes in the
MBN always have to be of type $n \to 1$.  However, one can show (see
Lemma~\ref{lem:represent-matrix}) that for each
$F_{k,b}$-matrix, there exists an MBN $(B', e')$ such that
$M_{e'}(B')$.  This can then be appended to $(B,e)$ which has the same
affect as appending a single node with the $F_{k,b}$-matrix.

\paragraph*{Simplifying MBNs to OBNs}

The characterisations of operations above ensure that updated MBNs
correctly evaluate to the updated probability distributions.  However,
rather than OBNs we obtain MBNs where the complexity of updates is
hidden in newly added nodes.  Evaluating such MBNs is computationally
more expensive because of the additional nodes. Below we show how to
simplify the MBN, minimising the number of nodes either after each
update or (in a lazy mode) after several updates.

As a first step we provide a lemma that will feature in all following
simplifications.  It states that every matrix can be
expressed by the composition of two matrices.
%
%

\begin{restatable}[Decomposition of matrices]{lem}{LemDecompMatrix}
  \label{lem:decomp-matrix}
  Given a matrix $P$ of type $n\to m$ and a set of $k < m$ outputs --
  without loss of generality we pick $\{ m-k+1, \dots, m \}$ -- there
  exist two matrices $P^\vdash: n \rightarrow m-k$ and
  $P^\dashv: n+m-k \rightarrow k$ such that
  \begin{equation}
    \label{eq:arc_reversal_condition}
    (\id_{m-k} \otimes P^\dashv) \cdot ((\nabla_{m-k} \cdot P^\vdash) 
    \otimes \id_n) \cdot \nabla_n = P,
  \end{equation}
  which is visualized in Figure~\ref{fig:arc_reversal}.  Moreover, the
  matrices can be chosen so that $P^\dashv$ is stochastic and
  $P^\vdash$ sub-stochastic.  If $P$ is stochastic $P^\vdash$ can be
  chosen to be stochastic as well.
\end{restatable}

We can now deduce the known special
case of arc reversal in OBN, stated e.g. in \cite{cb:arc-reversal}.

\begin{wrapfigure}{r}{0.45\textwidth}
	 \center
	 \vspace{-3ex}
	 \includegraphics[width=0.44\textwidth]{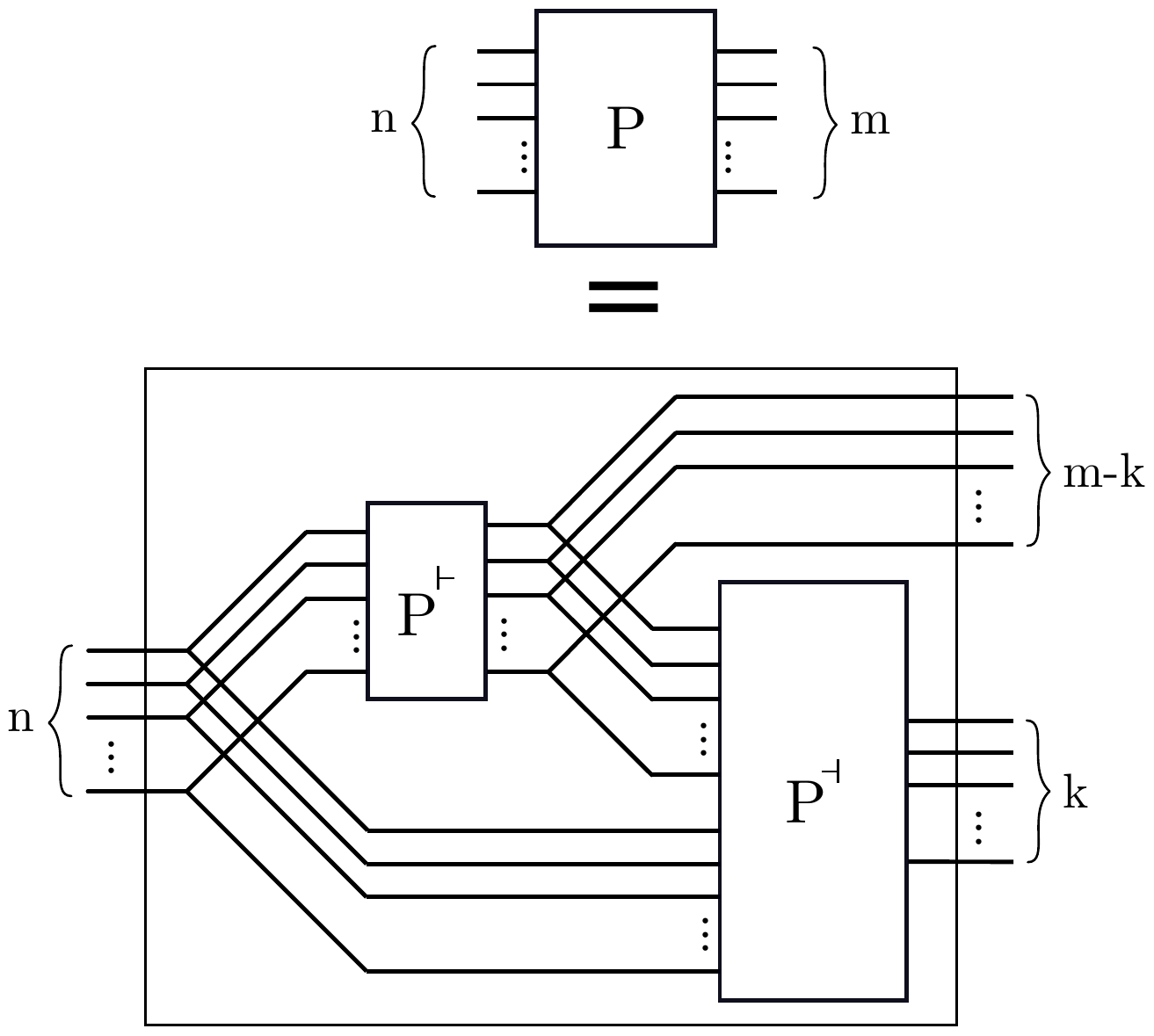}
	 \vspace{-5pt}  
	 \caption{Schematic string diagram depiction of the decomposition of matrices.}
	 \label{fig:arc_reversal}
\end{wrapfigure}
\begin{restatable}[Arc reversal in OBNs]{cor}{CorArcReversal}
  \label{cor:arc-reversal}
  Let $(B,e)$ be an OBN with $B = (V,\ell,s,\out)$ and two nodes
  $u, y \in V$, where $u$ is a direct predecessor of $y$, i.e.
  $u \in \mathrm{pred}(y)$.  Then there exists an OBN $(B', e')$ with
  $B' = (V,\ell',s',\out)$, evaluating to the same probability
  distribution, where $\ell'(v) = \ell(v)$, $s'(v) = s(v)$ if
  $v\neq u$ and $v\neq y$ and $y \in \mathrm{pred}(u)$.  Thus the
  dependency between $u$ and $v$ is reversed.
\end{restatable}


Arc reversal comes with a price: as can be seen in the proof, if $u$
is associated with a matrix $P_u\colon n\to 1$ and $y$ with a matrix
$P_y\colon m+1\to 1$, then we have to create new matrices
$P'_u\colon m+n+1\to 1$ and $P'_y\colon m+n\to 1$, causing new
dependencies and increasing the size of the matrix. Hence arc reversal
should be used sparingly.

After arc reversal a node might have duplicated inputs, which can be
resolved by multiplying the corresponding matrix with $\nabla$, thus
reducing the dimension.

Next, we can use Lemma~\ref{lem:decomp-matrix} to show that every
matrix can be represented as an MBN.  This MBN can always be built in a
``minimal'' way in that only $m$ nodes are needed to represent a
$n \to m$ matrix.

\begin{restatable}{lem}{LemRepresentMatrix}
  \label{lem:represent-matrix}
  Let $M: n \rightarrow m$ be a (sub-stochastic) matrix.  Then there
  exists an MBN $(B,e)$ with $B=(V,l,s,\mathrm{out})$ such that
  $M = M_{e}(B)$, $|V| = m$ and $\mathrm{out}$ is a bijection.
  Moreover, if $M$ is stochastic we can guarantee that $e(l(v))$ is
  stochastic for all $v \in V$.  If $M$ is sub-stochastic we can
  guarantee that $v_\mathit{front}$ -- the first node in a topological
  ordering of all nodes $V'$ -- is the only node where $e(l(v))$ is
  sub-stochastic, all other nodes have stochastic matrices.
\end{restatable}

\begin{cor}
  \label{cor:represent-matrix-stoch}
  Let $(B,e)$ be an MBN without inputs and assume that $M_e(B)$ is
  stochastic.  Then there exists an OBN $(B',e')$ such that
  $M_e(B) = M_{e'}(B')$.
\end{cor}

\begin{proof}
  The result follows trivially from the assumptions because for a
  stochastic MBN without input ports $M_e(B)$ is simply a column
  vector holding a probability distribution.  It is well known that
  every probability distribution can be represented by some (ordinary)
  Bayesian net.  Alternatively the result follows directly from
  Lemma~\ref{lem:represent-matrix}.
\end{proof}

We just argued that every MBN can be simplified so that it does not
contain any unnecessary nodes and at most one sub-stochastic matrix.
However, while Lemma~\ref{lem:represent-matrix} shows that these
simplifications are always possible it is not helpful in practice: in
fact in the proof we take the full matrix represented by an MBN and
then split it into (coupled) single nodes.  Since we chose to use MBNs
in order not to deal with large distribution vectors in the first
place, this approach is not practical.  Instead, in the following we
will describe methods which allow us to simplify an MBN without
computing the matrix first.

First note that MBNs stemming from CNU operations can contain
substructures that can locally be replaced by simpler ones.  They are
depicted in Figure~\ref{fig:simplifications}.
\begin{figure}
  \center
  \includegraphics[width=\textwidth]{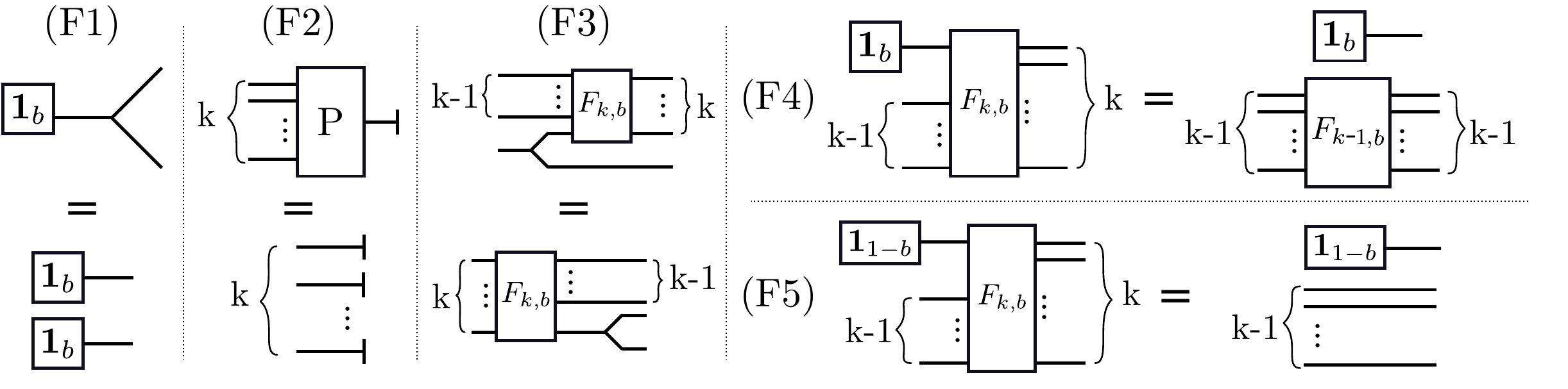}
  \caption{Equalities on sub-stochastic matrices. Note that (F2) holds
    only if $P$ is stochastic and for (F4) and (F5) we have to assume
    $k>1$.}
  \label{fig:simplifications}
\end{figure}

\begin{restatable}{lem}{LemEqualities}
  \label{lem:equalities}
  The equalities of Figure~\ref{fig:simplifications} hold for
  (sub-)stochastic matrices.
\end{restatable}

As a result, it makes sense to first eliminate all of these
substructures.  Then
there are two issues left to obtain an OBN.
First, there are nodes that lost their direct connection with an
output port (since output ports were terminated in a $\setop$
operation or since we added an $F_{k,b}$-matrix).  Those have to be
merged with other nodes.  Second, there are sub-stochastic matrices
that have to be eliminated as well.  The following lemma states that a
node not connected to output ports can be merged with its direct
successor nodes.  This can introduce new dependencies between these
successor nodes, but we remove one node from the network.

\begin{restatable}{lem}{LemEliminateNodeWithoutOutput}
  \label{lem:eliminate-node-without-output}
  Let $B = (V,\ell,s,\out)$ be a causality graph, $e$ an evaluation
  function such that $(B,e)$ is an MBN.  Assume that a node $v_0 \in V$
  is not connected to an output port, i.e. for all
  $i \in \{1, \dots, m\}: v_0 \neq \mathrm{out}(o_i)$, and
  $e(\ell(v_0))$ is stochastic.  Then there exists an MBN $(B',e')$
  with $B' = (V \backslash \{ v_0 \},\ell',s',\out)$ such that
  $M_e(B) = M_{e'}(B')$.  Moreover,
  $e'\circ \ell'|_{\bar{V}} = e\circ \ell|_{\bar{V}}$ and
  $s'|_{\bar{V}} = s|_{\bar{V}}$ where
  $\bar{V}=V\backslash (\{v_0\}\cup \mathrm{succ}(v_0))$.
\end{restatable}

The conditions on $\ell'$ and $s'$ mean that the update on $B$ is
local as it does not affect the whole network.  Only the direct
successors of $v_0$ are affected.

Finally, we have to get rid of sub-stochastic matrices inside the MBN,
which have been introduced by the $\assertop$ and $\nassertop$
operations (we assume that we did not normalize yet).  The idea is to
exchange nodes labelled with sub-stochastic matrices with the
predecessor nodes and move them to the front (as in
Lemma~\ref{lem:represent-matrix}). Once there, normalization is
straightforward by normalizing the vectors associated to these nodes.


\begin{restatable}{lem}{EliminateSubStoch}
  \label{lem:eliminate-sub-stoch}
  Let $B = (V,\ell,s,\out)$ be a causality graph without input ports,
  i.e. of type $0\to m$, $e$ an evaluation function such that $(B,e)$
  is an MBN. Furthermore we require that there is a one-to-one
  correspondence between output ports and nodes, i.e., $\out$ is a
  bijection.

  Assume that $V'\subseteq V$ is the set of all nodes equipped with
  sub-stochastic matrices, i.e.  $e(\ell(v))$ is sub-stochastic for
  all $v\in V'$.  Then there exists an OBN $(B',e')$ with
  $B' = (V,\ell',s',\out)$ such that $M_e(B) = M_{e'}(B')\cdot p_B$
  where $p_B = \top_m\cdot M_e(B) \le 1$ is the probability mass of
  $B$.  Moreover, $e'\circ \ell'|_{\bar{V} } = e\circ \ell|_{\bar{V}}$
  and $s'|_{\bar{V}} = s|_{\bar{V}}$ where
  $\bar{V} = V\backslash (V'\cup \mathrm{pred}^*(V'))$.
\end{restatable}

Note that $\frac{1}{p_B}$ (whenever $p_B\neq 0$) is the normalization
factor that can be obtained by terminating all input ports of $B$. We
do not have to compute $p_B$ explicitly, but it can be derived from
the probabilities of the nodes which have been moved to the front (see
proof).

\begin{restatable}{cor}{CorConstructObn}
  \label{cor:construct-obn}
  Let $B = (V,\ell,s,\out)$ be a causality graph without input ports,
  i.e. of type $0\to m$, $e$ an evaluation function such that $(B,e)$
  is an OBN. Let $P = M_e(B)$.

  Then we can construct OBNs representing
  $\setop_{A,b}(P), \assertop_{A,b}(P), \nassertop_{A,b}(P)$, where
  \begin{itemize}
  \item the $\setop$ operation modifies only
    $\{\out(o_i) \mid i\in A\}$ and their direct successors \emph{and}
  \item the $\assertop$ and $\nassertop$ operations modify only
    $\{\out(o_i) \mid i\in A\}$ and their predecessors.
  \end{itemize}
\end{restatable}

The operations are costly whenever a node has many predecessors or
direct successors. In a certain way this is unavoidable because our
operations are related to the computation of marginals, which is
$\mathsf{NP}$-hard \cite{cooper:marginals-np-hart}.  However, if the
Bayesian network has a comparatively ``flat'' structure, we expect
that the efficiency is rather high in the average case, as supported
by our runtime results below. Applying the $\nassertop$ operation will
introduce dependencies for the random variables corresponding to the
pre- and post-conditions of a transition, however this effect is
localized if we consider particular classes of Petri nets, such as
free-choice nets \cite{de:free-choice-petri}.


\begin{figure}
  \centering
  \includegraphics[width=\textwidth]{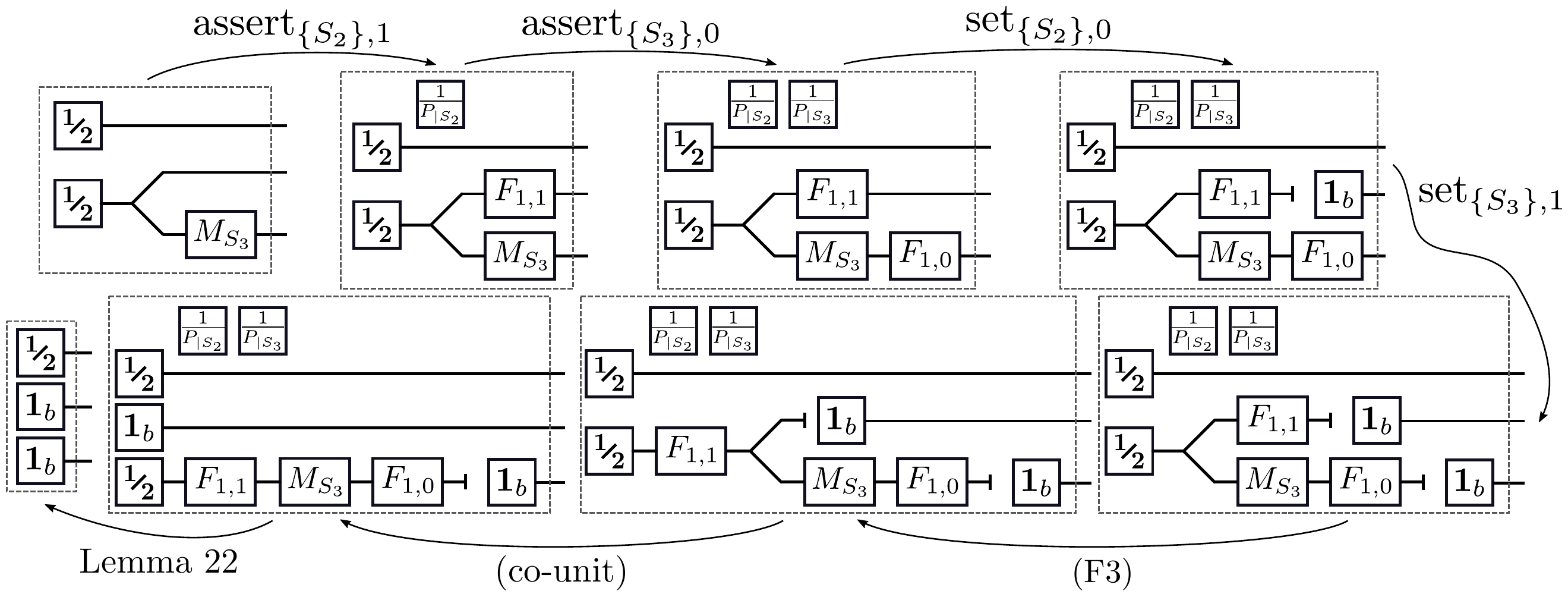}
  \caption{Exemplary update process for the $\successop_{t_4}$ operation of our running CNU example. Here $\sfrac{\mathbf{1}}{\mathbf{2}} = \begin{pmatrix}\sfrac{1}{2}\\ \sfrac{1}{2}\end{pmatrix}$ and $M_{S_3} = \begin{pmatrix}
  		\sfrac{1}{3} & \sfrac{1}{2} \\
  		\sfrac{2}{3} & \sfrac{1}{2} 
	\end{pmatrix}$.}
  \label{fig:update-example-mbn}
\end{figure}

\begin{ex}
  Figure~\ref{fig:update-example-mbn} shows an update process,
  following a lazy evaluation strategy, for a Bayesian net
  representing the probability distribution from
  Figure~\ref{fig:update-example}.
\end{ex}



\section{Implementation}
In order to quantitatively assess the performance of MBNs we developed a prototypical \verb+C+++ implementation of the concepts in this paper, allowing to read, write, simplify, generate, and visualize MBNs as well as perform operations on CNUs that update an underlying MBN.
The implementation is open-source and freely available on GitHub.\footnote{\url{https://github.com/bencabrera/bayesian_nets_program}}

As a first means of obtaining runtime results we randomly generated CNs with a range of different parameters: e.g. number of places, number of places in a precondition of a transition, places in the initial marking etc.
We then successively picked transitions at random to fire
and performed the necessary operations to update the MBN and simplify
it to an OBN.

We chose to guarantee a success rate of transition firing of around
$1/3$. We argue that given the fact that we model an observer with
prior knowledge it is realistic to assume a certain rate of successful
transitions. A very low sucess rate leads to an accumulation of
successive $F_{k,b}$ matrices which can only be eliminated using the
costly operations on substochastic matrices (see proof of Lemma
\ref{lem:eliminate-sub-stoch}).  One could implement effective
simplification strategies merging successive $F_{k,b}$ matrices --
since composing 0,1 diagonal matrices yields again 0,1 diagonal
matrices.  However, this is out of scope of this publication.


The plot on the left of Figure \ref{fig:runtime_results} shows a
comparison between run times when performing CNU operations directly
on the joint distribution versus our MBN implementation.  One can
clearly observe the exponential increase when using the joint
distribution while the MBN implementation in this setup stays
relatively constant.  The plot on the right of Figure
\ref{fig:runtime_results} hints towards an increase in complexity when
CNs -- and thus MBNs -- are more coupled.  When increasing the maximum
number of places in the precondition of a transition we observe an
increase in run times. The number of outliers with a dramatic increase
in run times seem to rise as well.

\begin{figure}
	\includegraphics[width=.49\textwidth]{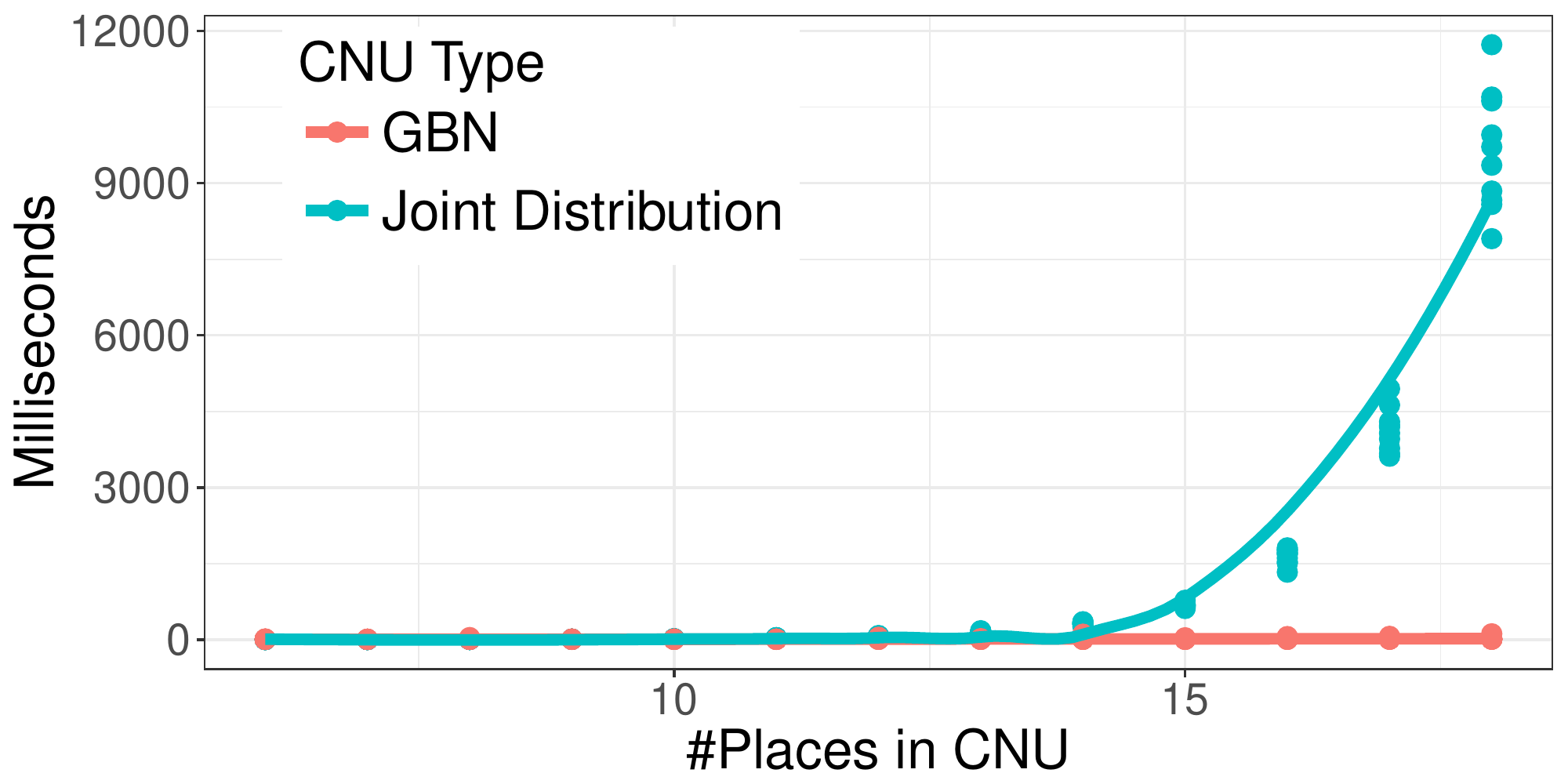}
	\hspace{0.6em}
	\includegraphics[width=.49\textwidth]{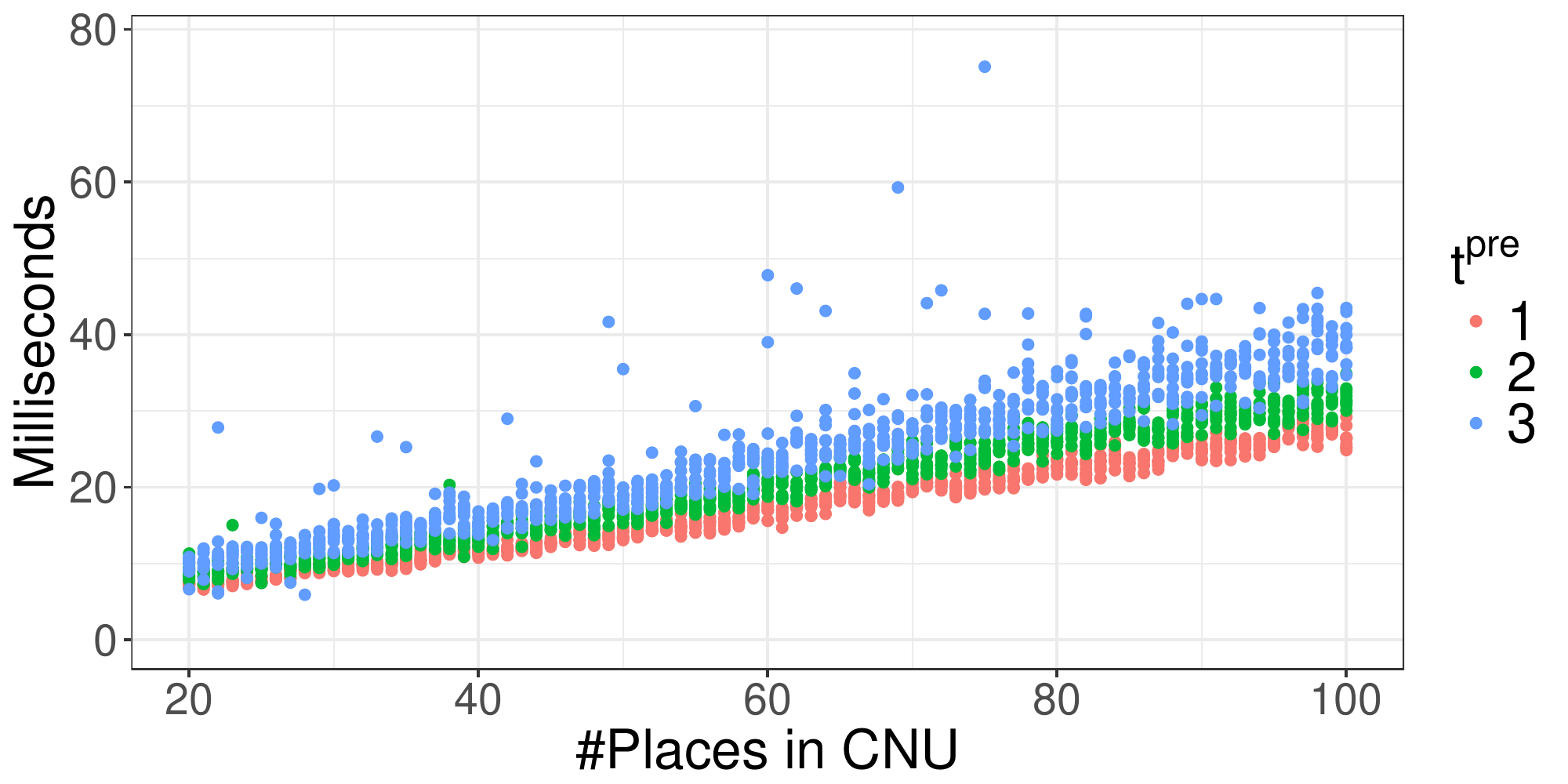}
	\caption{Averaged runtimes for performing 100 CNU operations using joint distributions or MBNs.}
	\label{fig:runtime_results}
\end{figure}

\section{Conclusion}

\noindent\textbf{Related work:} A concept similar to our nets with
uncertainty has been proposed in \cite{jr:uncertainty-initial-petri},
but without any mechanism for efficiently representing and updating
the probability distribution. There are also links to Hidden Markov
Models \cite{r:hidden-markov-models} for inferring probabilistic
knowledge on hidden states by observing a model.
\todo{
	\textbf{R1:} Lemma 18 reminded me a bit of Fritz's work on presenting the category of stochastic matrices (https://arxiv.org/abs/0902.2554), although now that I look at it again it's a fairly different setting. \\
	\textbf{Ben:} Since the reference focusses rather on the technical details of the stoch. mat PROP we could just cite it in the proof sketch after Prop. 9.
}

Bayesian networks were introduced by Pearl in
\cite{p:bayesian-networks} to graphically represent random variables
and their dependencies. Our work has some similarities to his
probabilistic calculus of actions (do-calculus)
\cite{p:probabilistic-calculus-actions} which supports the empirical
measurement of interventions. However, while Pearl's causal networks
model describe true causal relationships, in our case Bayesian
networks are just compact symbolic representations of huge probability
distributions.  There is also a notion of dynamic Bayesian networks
\cite{m:dynamic-bayesian-networks}, where a random variable has a
separate instance for each time slice.  We instead keep only one
instance of every random variable, but update the BN itself.  There is
substantial work on updating Bayesian networks (for instance
\cite{fg:sequential-update-bn}) with the orthogonal aim of learning
BNs from training data.

PROPs have been introduced in \cite{m:categorical-algebra},
foundations for term-based proofs have been studied in
\cite{b:languages-monoidal} and their graphical language has been
developed in \cite{s:graphical-monoidal,ck:picturing-quantum}.
Bayesian networks as PROPs have already been studied in
\cite{f:causal-theories} under the name of causal theories, as well as
in \cite{jz:influence-bayesian,jz:transformer-bayesian} in order to
give a predicate/state transformer semantics to Bayesian
networks. However, these papers do not explicitly represent the
underlying graph structure and in particular they do not consider
updates of Bayesian networks.

We use the results from \cite{cg:term-graphs-gs-monoidal} in order to
show that our causality graphs are in fact term graphs, which are
freely generated gs-monoidal categories, which in turn are
CC-structured PROPs. Although this result is intuitive, it is
non-trivial to show: given two terms with isomorphic underlying
graphs, each can be reduced to a normal form which can be converted
into each other using the axioms of a CC-structured PROP. Similar
results are given in \cite{fc:algebra-dags,bgm:normal-forms-journal}
for PROPs with multiplication and unit, in addition to
comultiplication and counit.

\smallskip

\noindent\textbf{Future work:} 
%
We would like to investigate further operations on probability
distributions, however it is unclear whether every operation can be
efficiently implemented. For instance linear combinations of
probability distributions seem difficult to handle.


Van der Aalst~\cite{Perpetual-Free-Choice} showed that all reachable
markings in certain free-choice nets can be inferred from their
enabled transitions. An unrestricted observer may therefore be in a
very strong position. Privacy research often considers statistical
queries, such as how many records with certain properties exist in the
database~\cite{DBLP:conf/tamc/Dwork08,DBLP:journals/geoinformatica/Damiani14}. To
model such weaker queries we require labelled nets where instead of
transitions we observe their labels. To implement this in BNs requires
a disjunction of the enabledness conditions of all transitions with
the same label. Furthermore we are interested in scenarios where
certain transitions are unobservable.

\todo{\textbf{Ben:} We should not forget to thank the reviewers}

\bibliography{references}


\full{
\appendix

\section{Proofs}
\label{apx:proofs}


\PropConditioned*

\begin{proof}
  For the case of $\failop^\mathit{pre}$ and $\failop^\mathit{post}$
  the equation is a straightforward reformulation of the
  definition. For $\successop$ we have:
  \begin{eqnarray*}
    && \prob(\markt{m}\mid \markt{}) =
    \frac{\prob(\markt{m})}{\prob(\markt{})} =
    \sum_{m'\overset{t}{\Rightarrow} m}
    \frac{\prob(\{m'\})}{\prob(\markt{})} =
    \sum_{m'\overset{t}{\Rightarrow} m} \prob(\{m'\}\mid
    \markt{}) \\
    & = & \sum_{m'\overset{t}{\Rightarrow} m}
    \assertop_{\rightdot{t},0}(\assertop_{\leftdot{t},1}(p)(m')) =
    \sum_{m=\bar{m}\cup\rightdot{t}} \sum_{\bar{m}=m'\backslash
      \leftdot{t}}
    \assertop_{\rightdot{t},0}(\assertop_{\leftdot{t},1}(p)(m')) \\
    & = & \sum_{m=\bar{m}\cup\rightdot{t}} \setop_{\leftdot
      t,0}(\assertop_{\rightdot{t},0}(\assertop_{\leftdot{t},1}(p)))(\bar{m}) \\
    & = & \setop_{\rightdot t,1}(\setop_{\leftdot
      t,0}(\assertop_{\rightdot{t},0}(\assertop_{\leftdot
      t,1}(p))))(m) = \successop_{t}(p)(m)
  \end{eqnarray*}
\end{proof}

\LemMbnSplitting*

\begin{proof}[Proof sketch]
  Let $V_1 = \mathrm{pred}^*(V')$ be the set of all predecessor nodes
  of $V'$ and let $V_3$ be the set of remaining nodes. Furthermore let
  $k = n+|V_1|$. Since $V'$ is path-closed we can construct a CG $B_2$
  that contains exactly the nodes in $V'$ and has as input ports
  exactly those ports needed by these nodes and one output port for
  every element in $V'$.

  Then we can construct a CG that contains all nodes of $V_1$ and
  whose output ports link to all input ports as well as all nodes in
  $V_1$. We then duplicate those wires that are needed by $B_2$ and
  permute them to the end of the output port sequence. This gives us
  the CG $B_1$.

  Finally, $B_3$ contains all nodes of $V_3$: due to the wiring it can
  access all input ports as well as all nodes of $V_1$ and $V_2$. At
  the very end all wires are terminated, duplicated and/or permuted as
  required by $B$.
\end{proof}

\begin{lem}
  \label{lem:mapping-m}
  For each MBN~\((B,e)\), %
  $M_e(B)$ is the value of the free extension of \(e\) %
  to CGs from Proposition~\ref{prop:free-cc-prop} applied to the
  isomorphism class of~\(B\) %
  using that it is actually the free CC-PROP %
  as it is the free GS-monoidal
  category \cite{cg:term-graphs-gs-monoidal}. 
\end{lem}

\begin{proof}
  It is sufficient to show that $M_e$ is functorial, in particular it
  respects composition and tensor, as well as identity, $\nabla$,
  $\sigma$ and $\top$. We only consider the following two cases, the
  rest is analogous.

  \medskip

  \hrule

  \medskip

  For instance, let two MBNs $B_i = (V_i,\ell_i,s_i,\mathit{out}_i)$,
  $i\in\{1,2\}$, with $B_1\colon n\to m$, $B_2\colon m\to \ell$ be
  given. We set $B = B_1;B_2$. We apply $M_e$ to $B$ and obtain for
  $\bitvec{x} \in \{0,1\}^\ell$, $\bitvec{z} \in\{0,1\}^m$:
  \begin{eqnarray*}
    &&
    \bigg(M_e(B_2)\cdot M_e(B_1)\bigg)(\bitvec{x},\bitvec{z}) \\
    & = & \sum_{\bitvec{y}\in\{0,1\}^m}
    M_e(B_2)(\bitvec{x},\bitvec{y}) \cdot
    M_e(B_2)(\bitvec{y},\bitvec{z}) \\
    & = & \sum_{\bitvec{y}\in\{0,1\}^m}
    \bigg(\sum_{\substack{b_2\colon
        W_{B_2}\to \{0,1\}\\ b_2(i_j) = y_j \\
        b_2(\textit{out}_2(o_i)) = x_i}} \quad \prod_{v\in V_2} \quad
    e(\ell_2(v))(b_2(v)\mid b_2(s_2(v)))\bigg)  \cdot\\
    && \qquad\qquad \bigg(\sum_{\substack{b_1\colon
        W_{B_1}\to \{0,1\}\\ b_1(i_j) = z_j \\
        b_1(\textit{out}_1(o_i)) = y_i}} \quad \prod_{v\in V_1} \quad
    e(\ell_1(v))(b_1(v)\mid b_1(s_1(v)))\bigg) \\
    & = & \sum_{\bitvec{y}\in\{0,1\}^m} \sum_{\substack{b_2\colon
        W_{B_2}\to \{0,1\}\\ b_2(i_j) = y_j \\
        b_2(\textit{out}_2(o_i)) = x_i}} \quad
    \sum_{\substack{b_1\colon
        W_{B_1}\to \{0,1\}\\ b_1(i_j) = z_j \\
        b_1(\textit{out}_1(o_i)) = y_i}} \\
    && \qquad\qquad \prod_{v\in V_2}\quad \prod_{v\in V_1} \quad
    \bigg(e(\ell_2(v))(b_2(v)\mid b_2(s_2(v))) \cdot
    e(\ell_1(v))(b_1(v)\mid b_1(s_1(v)))\bigg) \\
    & = & \sum_{\substack{b\colon
        W_B\to \{0,1\}\\ b(i_j) = z_j \\
        b(\textit{out}(o_i)) = x_i}} \qquad \prod_{v\in V}\quad
    e(\ell(v))(b(v)\mid b(s(v))) \\
    & = & M_e(B) = M_e(B_1;B_2)
  \end{eqnarray*}
  We assume that $\bitvec{x} = x_1\dots x_\ell$,
  $\bitvec{y} = y_1\dots y_m$, $\bitvec{z} = z_1\dots z_n$.

  Note that the equality sign on the second last line is due to the
  fact that assignments
  $b_1\colon W_{B_1}\to \{0,1\}, b_2\colon W_{B_2}\to \{0,1\}$ of
  boolean values to wires can be merged into one assignment
  $b\colon W_B\to \{0,1\}$ on $B$ whenever they agree on the
  interface, i.e., whenever
  $b_1(\mathit{out}_1(o_k)) = y_k = b_2(i_k)$.

  \medskip

  \hrule

  \medskip

  Next we check that $M_e(\sigma) = \sigma$. (Here we use some
  overloading: $\sigma$ stands for a CG as well as for a stochastic
  matrix.)  The MBN $\sigma$ is of the form
  $(V,\ell,s,\mathit{out}) = (\emptyset,[\,],[\,],[o_1\mapsto i_2,
  o_2\mapsto i_1])\colon 2\to 2$.

  Let $x_1,x_2,y_1,y_2\in\{0,1\}$. We compute
  \[
    M_e(\sigma)(x_1 x_2,y_1 y_2) = \sum_{\substack{b\colon
        W_\sigma\to \{0,1\}\\ b(i_j) = y_j \\
        b(\textit{out}(o_i)) = x_i}} \qquad \prod_{v\in V}\quad
    e(\ell(v))(b(v)\mid b(s(v)))
  \]
  In this case the product is always empty, evaluating to a value
  of~$1$. The sum is non-empty whenever an assignment $b$ exists,
  i.e., if $y_1 = b(i_1) = b(\mathit{out}(o_2)) = x_2$ and
  $y_2 = b(i_2) = b(\mathit{out}(o_1)) = x_1$. In these cases,
  $M_e(\sigma)(x_1 x_2,y_1 y_2) = 1$, otherwise the sum is empty and
  $M_e(\sigma)(x_1 x_2,y_1 y_2) = 0$. Combined, we obtain
  $M_e(\sigma) = \sigma$.
\end{proof}

\LemSetMatrix*

\begin{proof}
  First consider the special case of a singleton $A = \{1\}$ and
  $b=1$. (The case for $b=0$ is analogous.)  We compute
  \begin{equation*}
    M := \bigotimes_{i=1}^m T_{\{1\},1}^{\setop}(i) = 
    \begin{pmatrix}1 & 1 \\ 0 & 0\end{pmatrix} 
    \otimes \bigotimes_{i=2}^m \id = 
    \begin{pmatrix}1 & 1 \\ 0 & 0\end{pmatrix} \otimes \id_{m-1} = 
    \begin{pmatrix}\id_{m-1} & \id_{m-1} \\ 0_{m-1} & 0_{m-1}\end{pmatrix}.
  \end{equation*}
  Thus, whenever $\bitvec x \in \{0,1\}^m$:
  \begin{equation*}
    (M \cdot P)(\bitvec{x}) = \begin{cases}
      P(1\bitvec{x}_{[2\dots m]}) + P(0\bitvec{x}_{[2\dots m]}) & 
      \text{ if } \bitvec{x}_{[1]} = 1 \\
      0 & \text{ if } \bitvec{x}_{[1]} = 0
    \end{cases}
    \quad= \setop_{\{1\},1}(P).
  \end{equation*}
  The general case follows by using
  $\setop_{A,b} = \setop_{\{s_k\},b} \circ \dots \circ
  \setop_{\{s_1\},b}$ where $A = \{s_1,\dots,s_k\}$ and the fact that
  because the stochastic matrices form a PROP (see the first law
  (mixing of composition and tensor) in
  Table~\ref{tab:axioms-cc-prop}), we have
  $\bigotimes_{i=1}^m T_{A,b}^\setop(i) = \prod_{l=1}^k
  (\bigotimes_{i=1}^m T_{\{s_l\},b}^\setop(i))$.  Moreover,
  $\bigotimes_{i=1}^m T_{A,b}^\setop(i)$ is a stochastic matrix
  because all $T_{A,b}^\setop(i)$ are stochastic and the tensor
  preserves this property.
\end{proof}

\LemAssertMatrix*

\begin{proof}
  Again we first consider the singleton case $A = \{1\}$ and
  $b=1$. (The case for $b=0$ is analogous.)  Then
  \begin{equation*}
    M := \bigotimes_{i=1}^m T_{\{1\},1}^\assertop(i) = 
    \begin{pmatrix}1 & 0 \\ 0 & 0\end{pmatrix} 
    \otimes \id_{m-1} = 
    \begin{pmatrix}\id_{m-1} & 0_{m-1} \\ 0_{m-1} & 0_{m-1}\end{pmatrix},
  \end{equation*}
  and thus $(M \cdot P)(\bitvec{x}) = \begin{cases}
    P(\bitvec{x}) & \text{ if } \bitvec{x}_{[1]} = 1 \\
    0 & \text{ if } \bitvec{x}_{[1]} = 0
  \end{cases}$ 

  Also
  $\left(\bigotimes_{i=1}^m Q_A(i)\right)(\bitvec x) = \begin{cases}
    1 & \text{ if } \bitvec{x}_{[A]} = \{1\} \\
    0 & \text{ otherwise }
  \end{cases}$ 

  As a result,
  $P_{|A} = \sum_{\bitvec x | \bitvec{x}_{[A]}=\{1\}} P(\bitvec x)$ and
  thus $\frac{1}{P_{|A}}(M \cdot P) = \assertop_{A,1}(P)$.  $M$ is not
  stochastic because clearly the last $m-1$ columns add up to $0$.

 Similarly to the $\setop$ case we have
  $\assertop_{A,b} = \assertop_{\{s_k\},b} \circ \dots \circ
  \assertop_{\{s_1\},b}$ where $A = \{s_1,\dots,s_k\}$ and
  $\bigotimes_{i=1}^m T_{A,b}^\assertop(i) = \prod_{l=1}^k
  (\bigotimes_{i=1}^m T_{\{s_l\},b}^\assertop(i))$.
\end{proof}

\LemNassertMatrix*

\begin{proof}
  We have
  $M := \left(F_{k,1} \otimes \id_{m-k}\right)
  = \begin{pmatrix}
    0_{m-k} & 0 \\
    0 & \id_k
  \end{pmatrix}$.  This means that when multiplying we get
  $(M \cdot P)(\bitvec{x}) = \begin{cases}
    0 & \text{ if } \bitvec{x}_{A} = \{1\} \\
    P(\bitvec{x}) & \text{ otherwise }
  \end{cases}$.  

  As shown in the $\assertop$ case
  $P_{|A} = \sum_{\bitvec x | \bitvec{x}_{[A]}=\{1\}} P(\bitvec x)$ and
  thus
  $P_{|A}^c = 1-P_{|A} = \sum_{\bitvec x | \bitvec{x}_{[A]} \not=1}
  P(\bitvec x)$ and together we get
  $\nassertop_{A,1}(P) = \frac{1}{P_{|A}^c} (M \cdot P)$.
\end{proof}

\LemDecompMatrix*

\begin{proof}
  In the following we denote by $\bitvec{z} \in \{0,1\}^n$,
  $\bitvec{x} \in \{0,1\}^{m-k}$ and $\bitvec y \in \{0,1\}^k$ bit
  vectors that represent the inputs of $P$ (for $\bitvec z$), outputs
  of $P^\vdash$ (for $\bitvec x$) respectively the outputs of
  $P^\dashv$ (for $\bitvec y$).  We now define
  \begin{equation*}
    P^\vdash(\bitvec x\mid\bitvec z) = \sum_{\bitvec v \in
      \{0,1\}^k} P(\bitvec x \bitvec v \mid \bitvec z) \quad
    \text{ and } \quad P^\dashv(\bitvec y \mid \bitvec x
    \bitvec z) = \frac{P(\bitvec x \bitvec y \mid \bitvec
      z)}{\sum_{\bitvec v \in \{0,1\}^k} P(\bitvec
      x \bitvec v \mid \bitvec z)}
    = \frac{P(\bitvec x \bitvec y \mid \bitvec
      z)}{P^\vdash(\bitvec x\mid\bitvec z)}
  \end{equation*}
  Whenever
  $P^\vdash(\bitvec x\mid\bitvec z) = \sum_{\bitvec v \in \{0,1\}^k}
  P(\bitvec x \bitvec v \mid \bitvec z) = 0$ it holds that
  $P(\bitvec x \bitvec v \mid \bitvec z) = 0$ for every
  $\bitvec v \in \{0,1\}^k$. In this case the row of $P^\dashv$
  corresponding to $\bitvec x \bitvec z$ can be chosen arbitrarily, as
  long as it adds up to $1$.  We observe that $P^\dashv$ is stochastic
  because
  $\sum_{\bitvec y_1 \in \{0,1\}^k} P^\dashv(\bitvec y_1| \bitvec x
  \bitvec z) = 1$ for all $\bitvec x$, $\bitvec z$.  $P^\vdash$ on the
  other hand is by definition stochastic if and only if $P$ is
  stochastic: If we keep the column index ($\bitvec x$ in the case of
  $P^\vdash$ and $\bitvec x\bitvec y$ in the case of $P^\dashv$) fixed
  and sum over the row index, we straightforwardly obtain $1$ in both
  cases.

  To use the more intuitive matrix notation one could equivalently
  define $P^\vdash = (\id_k \otimes \top_{m-k}) \cdot P$.  However,
  such a simple characterization in terms of matrix compositions does
  not exist for $P^\dashv$.

  Finally, we can check that $P^\vdash$ and $P^\dashv$ satisfy
  \eqref{eq:arc_reversal_condition}.
  \begin{equation*}
    \begin{aligned}
      &((\id_{m-k} \otimes P^\dashv) \cdot ((\nabla_{m-k} \cdot P^\vdash) \otimes \id_n) \cdot \nabla_n)(\bitvec x \bitvec y\mid\bitvec z) \\
      &= ((\id_{m-k} \otimes P^\dashv) \cdot ((\nabla_{m-k} \cdot P^\vdash) \otimes \id_n))(\bitvec x \bitvec y\mid\bitvec z \bitvec z) \\
      &= \sum_{\bitvec{x}_1, \bitvec{x}_2\in
        \{0,1\}^{m-k}, \bitvec{z}_1 \in \{0,1\}^n} (\id_{m-k} \otimes P^\dashv) \left(\bitvec x \bitvec y\mid \bitvec{x}_1 \bitvec{x}_2 \bitvec{z}_1\right)\cdot ((\nabla_{m-k} \cdot P^\vdash) \otimes \id_n)(\bitvec{x}_1 \bitvec{x}_2 \bitvec{z}_1\mid\bitvec z \bitvec z) \\
      &= \sum_{\bitvec{x}_1, \bitvec{x}_2 \in
        \{0,1\}^{m-k}, \bitvec{z}_1 \in \{0,1\}^n} \id_{m-k}(\bitvec x\mid\bitvec{x}_1) \cdot P^\dashv(\bitvec y\mid \bitvec{x}_2 \bitvec{z}_1) \cdot (\nabla_{m-k} \cdot P^\vdash) (\bitvec{x}_1 \bitvec{x}_2\mid\bitvec z) \cdot \id_n(\bitvec{z}_1\mid\bitvec z) \\
      &= \sum_{\bitvec{x}_2\in \{0,1\}^{m-k}} P^\dashv(\bitvec y\mid \bitvec{x}_2 \bitvec{z}) \cdot (\nabla_{m-k} \cdot P^\vdash) (\bitvec{x} \bitvec{x}_2\mid\bitvec z) \\
      &= \sum_{\bitvec{x}_2\in \{0,1\}^{m-k}} P^\dashv(\bitvec y\mid \bitvec{x}_2 \bitvec{z}) \cdot \bigg( \sum_{\bitvec{x}_3 \in \{0,1\}^{m-k}} \nabla_{m-k}(\bitvec{x} \bitvec{x}_2\mid \bitvec{x}_3) \cdot P^\vdash (\bitvec{x}_3\mid\bitvec z) \bigg) \\
      &= P^\dashv(\bitvec y\mid \bitvec{x} \bitvec{z}) \cdot P^\vdash (\bitvec{x}\mid\bitvec z) \\
      &= P(\bitvec x \bitvec y\mid\bitvec z)
    \end{aligned}	
  \end{equation*}
  where we used that $\id_i(\bitvec y\mid \bitvec x)$ is non-zero only
  when $\bitvec y=\bitvec x$ and that
  $\nabla_i(\bitvec{y}_1 \bitvec{y}_2\mid \bitvec{x})$ is non-zero iff
  $\bitvec{y}_1 = \bitvec{y}_2 = \bitvec{x}$.

  Whenever $P^\vdash(\bitvec x\mid\bitvec z) = 0$ the product in the
  second-last line is $0$. As argued above in this case
  $P(\bitvec x \bitvec y \mid \bitvec z) = 0$ as well and the last
  equality holds.
\end{proof}

\CorArcReversal*

\begin{proof}
  Let $g_u = \ell(u)$, $g_y = \ell(y)$ and $e(g_u) = P_u$,
  $e(g_y) = P_y$.
  
  We assume without loss of generality that $s(y) = \bitvec{u} u$,
  otherwise we have to rearrange the inputs of $y$. Let
  $m = |\bitvec{u}|$ and $n = |s(u)|$. 

  Since the set $\{u,y\}$ is closed with respect to paths, we can use
  Lemma~\ref{lem:mbn_splitting} and represent $B$ by the following
  term:
  $t_1;(\id_k\otimes (\underbrace{(\id_m\otimes
    (B_{g_u};\nabla));(B_{g_y}\otimes \id)}_{t_{uy}}));t_3$. The term
  $t_{uy}$ corresponds to the matrix
  $M_e(t_{uy}) = P = (\id\otimes P_y)\cdot ((\nabla\cdot P_u)\otimes
  \id_m)$, which is a matrix of type $m+n\to 2$. We now apply
  Lemma~\ref{lem:decomp-matrix} to $P$ for $k=1$ and obtain matrices
  $P'_y\colon m+n\to 1$, $P'_u\colon m+n+1\to 1$ with
  $(\id \otimes P'_u) \cdot ((\nabla \cdot P'_y) \otimes \id_m) \cdot
  \nabla_m = P$. We transform this into a term
  $t_{y} = \nabla_m;((B_{g'_y};\nabla) \otimes \id_m);(\id\otimes
  B_{g'_u})$ where $g'_y,g'_u$ are two new generators, evaluating to
  $P'_u$, $P'_y$.

  If we replace $t_{uy}$
  in the term above by $t_{yu}$ we obtain a new BN where the order of
  $u,y$ is reversed and the other structure remains unchanged.
\end{proof}

\LemRepresentMatrix*

\begin{proof}
  The statement is a direct result of Lemma~\ref{lem:decomp-matrix}.
  Using $k=1$ we get from Lemma~\ref{lem:decomp-matrix}
  $M_{m-1} = M^\vdash: n \rightarrow m-1$ and
  $M^{(m-1)} = M^\dashv: n+1 \rightarrow 1$ such that
  \begin{equation} 
    \label{eq:arc_reversal_condition}
    (\id \otimes M_{m-1}) \cdot ((\nabla \cdot M^{(m-1)}) \otimes \id_n) \cdot \nabla_n = M.
  \end{equation}
  We can now apply Lemma~\ref{lem:decomp-matrix} again to $M_{m-1}$ to
  get a new $M_{m-2}$ and $M^{(m-2)}$.  Doing this process recursively
  in step $i$ we end up with smaller and smaller matrices
  $M_{m-i}: n \rightarrow m-1$ and $M^{(m-i)}: n+1 \rightarrow 1$.
  After a total of $m-1$ steps we end up with $M_1: n \rightarrow 1$
  and we stop.  The matrices $M_1, M^{(2)}, \dots, M^{(m-1)}$ all have
  type $n+1 \rightarrow 1$.  Moreover, Lemma~\ref{lem:decomp-matrix}
  guarantees that $M^{(2)}, \dots, M^{(m-1)}$ can always be chosen to
  be stochastic matrices.  $M_1$ can be chosen stochastic if and only
  if $M$ is stochastic.  For $B$ we now set
  $V = \{ v_1, \dots, v_m\}$, $l(v_j) = g_j$ for all $j = 1,\dots,m$
  and $s(v_j) = (v_{j-1}, i_1, \dots, i_n)$ for $j = 2,\dots, m$ and
  $s(v_1) = (i_1, \dots, i_n)$ and $\mathrm{out}(o_j) = v_j$ for
  $j=1,\dots,m$.  Accordingly we set $e(g_j) = M^{(j)}$ for
  $j=2,\dots,m$ and $e(g_1) = M_1$.  It is easy to verify that $(B,e)$
  now forms an MBN.
\end{proof}

\LemEqualities*
\begin{proof}
  In the following we assume $b=1$. 
  The case $b=0$ is always analogous.
  To show (F1) we compute
  \begin{equation*}
    \nabla \cdot \begin{pmatrix}1 \\ 0\end{pmatrix} = \begin{pmatrix}1 \\ 0 \\ 0 \\ 0\end{pmatrix} = \begin{pmatrix}1 \\ 0\end{pmatrix} \otimes \begin{pmatrix}1 \\ 0\end{pmatrix}.
  \end{equation*}
  For (F2) let $P: k \to 1$ be a stochastic matrix. Then
  \begin{equation*}
    \top \cdot P = \begin{pmatrix}
      1 & 1
    \end{pmatrix} \cdot
    \begin{pmatrix}
      p_1 & \cdots & p_{2^{k}} \\
      \bar{p}_1 & \cdots & \bar{p}_{2^{k}}
    \end{pmatrix}
    =
    \underbrace{
      \begin{pmatrix}
        1 & \cdots & 1
      \end{pmatrix}}_{2^k \text{ times }}
    = \top_k.
  \end{equation*}
  For (F4) we calculate
  \begin{equation*}
    F_{k,1} \cdot (\begin{pmatrix}1\\0\end{pmatrix} \otimes \id_{k-1,1}) = F_{k,1} \cdot \begin{pmatrix}
      \id_{k-1} \\
      0_{k-1}
    \end{pmatrix} 
    = \begin{pmatrix}
      F_{k-1,1} \\
      0_{k-1}
    \end{pmatrix} 
    = \begin{pmatrix}1\\0\end{pmatrix} \otimes F_{k-1,1}.
  \end{equation*}
  To show (F5) we calculate
  \begin{equation*}
    F_{k,1} \cdot (\begin{pmatrix}0\\1\end{pmatrix} \otimes \id_{k-1}) = F_{k,1} \cdot \begin{pmatrix}
      0_{k-1} \\
      \id_{k-1} 
    \end{pmatrix} 
    = \begin{pmatrix}
      0_{k-1} \\
      \id_{k-1}
    \end{pmatrix} 
    = \begin{pmatrix}0\\1\end{pmatrix} \otimes \id_{k-1}.
  \end{equation*}
  In order to prove equality (F3) we have to show
  \begin{equation*}
    (F_{k,1} \otimes \id) \cdot (\id_{k-1} \otimes \nabla) =
    (\id_{k-1} \otimes \nabla) \cdot F_{k,1}.
  \end{equation*}
  Given $\bitvec x, \bitvec y\in\{0,1\}^k$ it holds that
  $F_{k,1}(\bitvec x, \bitvec y)  = 1$ iff $\bitvec x = \bitvec
  y$ and $\bitvec x, \bitvec y\neq 1\ldots 1$. Otherwise
  $F_{k,1}(\bitvec x, \bitvec y)  = 0$.

  Now, given
  $\bitvec{z}\mathrm{z}\in\{0,1\}^{k+1},\bitvec x \in
  \{0,1\}^k$, we have
  \begin{eqnarray*}
    && ((F_{k,1} \otimes \id) \cdot (\id_{k-1} \otimes
    \nabla))(\bitvec{z}\mathrm{z}\mid \bitvec x) \\
    & = & \sum_{\bitvec{y}\mathrm{y}\in\{0,1\}^{k+1}} (F_{k,1} \otimes
    \id)(\bitvec{z}\mathrm{z}\mid \bitvec{y}\mathrm{y}) \cdot (\id_{k-1}
    \otimes \nabla)(\bitvec{y}\mathrm{y}\mid \bitvec x) \\
    & = & \sum_{\bitvec{y}\mathrm{y}\in\{0,1\}^{k+1}} F_{k,1}(\bitvec
    z\mid \bitvec{y}) \cdot [\mathrm{z}=\mathrm{y}]\cdot [\bitvec
    x=\bitvec y]\cdot [\mathrm{y} = \bitvec{y}_{[k]} =
    \bitvec{x}_{[k]}] \\
    & = & \left\{
      \begin{array}{ll}
        F_{k,1}(\bitvec z\mid \bitvec x) & \text{if
          $\bitvec{x}_{[k]} = \mathrm{z}$} \\
        0 & \text{otherwise}
      \end{array}\right\}
    = \left\{
      \begin{array}{ll}
        1 & \text{if $\bitvec{z}=\bitvec{x}$,
          $\bitvec{x},\bitvec{z}\neq 1\ldots 1$,
          $\bitvec{x}_{[k]} = \mathrm{z}$} \\
        0 & \text{otherwise}
      \end{array}\right.
  \end{eqnarray*}
  Note that $[\bitvec{x}=\bitvec{y}]$ stands for $1$ if the
  equality holds and for zero otherwise.

  Furthermore:
  \begin{eqnarray*}
    && ((\id_{k-1} \otimes \nabla) \cdot
    F_{k,1})(\bitvec{z}\mathrm{z}\mid \bitvec x) \\
    & = & \sum_{\bitvec{y}\in\{0,1\}^{k}} (\id_{k-1} \otimes
    \nabla)(\bitvec{z}\mathrm{z}\mid \bitvec{y}) \cdot
    F_{k,1}(\bitvec{y}\mid \bitvec{x}) \\
    & = & \sum_{\bitvec{y}\in\{0,1\}^{k}}
    [\bitvec z=\bitvec y]\cdot [\mathrm{z} = \bitvec{z}_{[k]} =
    \bitvec{y}_{[k]}] \cdot F_{k,1}(\bitvec y \mid \bitvec{x})  \\
    & = & \left\{
      \begin{array}{ll}
        F_{k,1}(\bitvec z\mid \bitvec x) & \text{if
          $\bitvec{z}_{[k]} = \mathrm{z}$} \\
        0 & \text{otherwise}
      \end{array}\right\}
    = \left\{
      \begin{array}{ll}
        1 & \text{if $\bitvec{z}=\bitvec{x}$,
          $\bitvec{x},\bitvec{z}\neq 1\ldots 1$,
          $\bitvec{z}_{[k]} = \mathrm{z}$} \\
        0 & \text{otherwise}
      \end{array}\right.
  \end{eqnarray*}
  And it is easy to see that both end results are equal.

\end{proof}
\todo{\textbf{R3:} the material on page 23 deserved much more space and attention -- It is really illustrating the core points of your paper!}

\LemEliminateNodeWithoutOutput*

\begin{proof}
  We set $V' = \mathrm{succ}(v_0)$ and fix a topological ordering on
  $V'$. Then we successively exchange $v_0$ with the next successor in
  the topological ordering, using arc reversal as described in
  Corollary~\ref{cor:arc-reversal}. Note that in arc reversal the
  number of successors of $v_0$ decreases by one and the matrix
  associated to $v_0$ will remain stochastic (see
  Lemma~\ref{lem:decomp-matrix}), hence at some point $v_0$ will have
  no successors and we can use equality~(F2) from
  Figure~\ref{fig:simplifications} in order to eliminate it.

  Note that only the source and labelling functions of the direct
  successors of $v_0$ are affected and the respective functions remain
  unchanged for the nodes in $\bar{V}$.

\end{proof}

\EliminateSubStoch*

\begin{proof}
  We iterate over $V'$ and by using again
  Lemma~\ref{lem:represent-matrix} we can replace the sub-MBN induced
  by $v_0\in V'$ and its predecessors by one that has a single
  sub-stochastic matrix in the front, without predecessors.  Doing
  this iteratively, we can move every sub-stochastic matrix to the
  front of the MBN where they do not have any predecessors. This
  results in an MBN $(B',\hat{e})$. Only the nodes in $V'$ and their
  predecessor nodes, but not the other nodes are affected, that is
  $e'\circ \ell'|_{\bar{V}} = e\circ \ell|_{\bar{V}}$ and
  $s'|_{\bar{V}} = s|_{\bar{V}}$.

  Through normalization we can then get rid of the sub-stochasticity
  for every node: assume that $\hat{V}$ contains the nodes in $B'$
  equipped with sub-stochastic matrices. We terminate all output ports
  of $B'$ by computing
  $\top_m\cdot M_{\hat{e}}(B') = \top_m\cdot M_e(B) = p_B$, since
  $B,B'$ specify the same matrix. Due to equality (F2) in
  Figure~\ref{fig:simplifications} this means that all stochastic
  nodes disappear, only the sub-stochastic nodes remain and hence
  $p_B = \prod_{v\in\hat{V}} (Q_v(0)+Q_v(1))$ where
  $Q_v = \hat{e}(\ell'(v))$. Note that $Q_v$ is simply a column vector
  with two entries. The value $q_v = Q_v(0)+Q_v(1)$ results when a
  sub-stochastic node without predecessors is terminated. We now
  replace each $Q_v$ by $\frac{1}{q_v} Q_v$, which is a stochastic
  matrix, resulting in a new evaluation function $e'$. Looking at the
  definition of $M_e$ in Section~\ref{sec:general_setup}, we observe
  that the values $\frac{1}{q_v}$ can be factored out and hence:
  \begin{equation*}
    M_{e'}(B') = M_e(B)\cdot \prod_{v\in\hat{V}} \frac{1}{q_v} =
    M_e(B)\cdot \frac{1}{p_B}
  \end{equation*}
  If $q_v$ for some $v$ it holds that $p_B = 0$. In this case
  normalization is not possible and we set $B' = B$, but the result
  still holds since $M_e(B) = 0 = M'_{e'}(B')\cdot p_B$.
\end{proof}

\CorConstructObn*

\begin{proof}
  In a $\setop$ operation, we terminate the output ports of all nodes
  in $\{\out(o_i) \mid i\in A\}$ (Lemma~\ref{lem:set_matrix}). Then we
  have an MBN consisting only of stochastic matrices and with
  Lemma~\ref{lem:eliminate-node-without-output} we can convert the
  resulting net into an MBN, affecting only the direct successors of
  these nodes.

  In an $\assertop$ or $\nassertop$ operation (see
  Lemmas~\ref{lem:assert_matrix} and~\ref{lem:nassert_matrix}) we add
  sub-stochastic matrices $F_{k,b}$ and equality (F3) from
  Figure~\ref{fig:simplifications} allows us the shift these
  sub-stochastic matrices in such a way that all successors of
  predecessors of $F_{k,b}$ come behind $F_{k,b}$. The matrix
  $F_{k,b}$ can now be fused with its direct predecessors, using
  Lemma~\ref{lem:decomp-matrix}, resulting either in one
  sub-stochastic matrix (case $\assertop$) or in several
  sub-stochastic matrices (case $\nassertop$).

  Now we have eliminated all nodes not connected to output ports. We
  can also assume that two different output ports link to different
  nodes, since none of our operations introduces duplication.

  Then we can apply Lemma~\ref{lem:eliminate-sub-stoch} to eliminate
  the remaining sub-stochastic matrices and to normalize.
\end{proof}
}

\end{document}